\documentclass[11pt]{article}%
\usepackage{amssymb}
\usepackage{amsfonts}
\usepackage{amsmath}%
\usepackage{graphicx}
\newtheorem{theorem}{Theorem}

\newtheorem{definition}[theorem]{Definition}

\newtheorem{lemma}[theorem]{Lemma}

\newtheorem{proposition}[theorem]{Proposition}
\newtheorem{remark}[theorem]{Remark}

\newenvironment{proof}[1][Proof]{\noindent\textbf{#1.} }{\ \rule{0.5em}{0.5em}}
\begin{document}

\title{Hylomorphic Vortices in Abelian Gauge Theories}
\author{Vieri Benci$^{\ast}$, Donato Fortunato$^{\ast\ast}$\\$^{\ast}$Dipartimento di Matematica Applicata \textquotedblleft U.
Dini\textquotedblright\\Universit\`{a} di Pisa\\via F. Buonarroti 1/c 56127 Pisa, Italy\\e-mail: benci@dma.unipi.it\\$^{\ast\ast}$Dipartimento di Matematica \\Universit\`{a} di Bari and INFN sezione di Bari\\Via Orabona 4, 70125 Bari, Italy\\e-mail: fortunat@dm.uniba.it}
\date{}
\maketitle

\begin{abstract}
We consider an Abelian Gauge Theory in $\mathbb{R}^{4}$ equipped with the
Minkowski metric. This theory leads to a system of equations, the
Klein-Gordon-Maxwell equations, which provide models for the interaction
between the electromagnetic field and matter. We assume that the nonlinear
term is such that the energy functional is positive; this fact makes the
theory more suitable for physical models.

A three dimensional vortex is a finite energy, stationary solution of these
equations such that the matter field has nontrivial angular momentum and the
magnetic field looks like the field created by a finite solenoid. Under
suitable assumptions, we prove the existence of three dimensional vortex-solutions.

\end{abstract}
\tableofcontents

\section{Introduction\label{ded}}

Roughly speaking, a \textit{vortex} is a \textit{solitary wave} $\psi$\ with
non-vanishing angular momentum ($\mathbf{M}\left(  \psi\right)  \neq0$). A
\textit{solitary wave} is a solution of a field equation whose energy is
localized and which preserves this localization in time.

Here we are interested in proving the existence of vortices in Abelian gauge
theories. Abelian gauge theories, in $\mathbb{R}^{4}$ equipped with the
Minkowski metric, provide models for the interaction between the
electromagnetic field and matter. Actually an Abelian gauge theory leads to a
system of equations (see (\ref{e1}), (\ref{e2}), (\ref{e3})), the
Klein-Gordon-Maxwell equations (KGM), which occur in various physical problems
such as elementary particles, superconductivity, cosmology, nonlinear optics
(see e.g. \cite{rub}, \cite{fel}, \cite{raj}, \cite{vil}).

The KGM can be regarded as a perturbation of the nonlinear Klein-Gordon
equation (KG) (see (\ref{KG})).

So first we recall some existence results for KG:

\begin{itemize}
\item For the case $\mathbf{M}$ $\left(  \psi\right)  =0,$ we recall the
pioneering paper of Rosen \cite{rosen68} and \cite{Coleman78}, \cite{strauss},
\cite{Beres-Lions}. When the lower order term $W$ $\geq0$ (see (\ref{KG})),
the spherically symmetric solitary waves have been called $Q$ -balls by
Coleman in \cite{Coleman86} and this is the name used in the physics literature.

\item Vortices for KG in two space dimensions have been investigated in
\cite{Kim93}; later also three dimensional vortices for KG have been
investigated (see \cite{Be-Visc}, \cite{bad}).
\end{itemize}

Now let us see some literature on KGM. We notice that the peculiarities of the
model depend on the lower order term $W$ and it is relevant to distinguish
various situations.

\begin{itemize}
\item For the case $\mathbf{M}$ $\left(  \psi\right)  =0,$ the existence of
solitary waves for KGM was first proved in \cite{bf} assuming that
\begin{equation}
W(s)=\frac{1}{2}s^{2}-\frac{s^{p}}{p},\text{ }4<p<6,\text{ }s\geq0.\label{oc}%
\end{equation}
The existence of solitary waves for KGM in this situation (i.e. with
$\mathbf{M}\left(  \psi\right)  =0$ and $W$ as in (\ref{oc})) has been studied
also in \cite{ca}, \cite{tea}, \cite{tea2}, \cite{dav}. In these papers the
existence and the non-existence of stationary solutions has been proved under
different assumptions.
\end{itemize}

However the lower order term $W$ defined by (\ref{oc}) is not suitable to
model interesting physical models since in this case there are configurations
with negative energy and the evolution problems relative to KGM does not
possess in general global solutions (cf. e. g. \cite{befogranas}). So the
request
\[
W\geq0
\]
seems to be necessary to get solutions which are physically meaningful.

\begin{itemize}
\item The case $W\geq0$ and $\mathbf{M}\left(  \psi\right)  =0$ has been
treated in \cite{befo}.
\end{itemize}

Now let us consider the existence of vortices ($\mathbf{M}\left(  \psi\right)
\neq0$) for KGM.

\begin{itemize}
\item The existence of vortices for Abelian gauge theories in two space
dimensions has been discovered in a seminal paper by Abrikosov \cite{ab} in
the study of the superconductivity. Then, in \cite{nil}, the planar vortices
are studied in the context of elementary particles (see also the books
\cite{fel}, \cite{raj}, \cite{rub}, \cite{yangL} with their references). We
point out that, in these cases, the function $W$ that has been considered is
of the type
\begin{equation}
W(s)=\left(  1-s^{2}\right)  ^{2}\label{v1}%
\end{equation}
namely it is a double well shaped and positive function. However, if
(\ref{v1}) holds, there are not vortices in 3 space dimensions for KGM (see
Theorem \ref{b}).

\item In \cite{befov07}, \cite{bis} the existence of vortices in 3 space
dimensions has been proved assuming (\ref{oc}).
\end{itemize}

The aim of this paper is to to prove the existence of vortices in 3 space
dimensions also when $W(s)\geq0\;$and $W(0)=0.$ More precisely we consider the
term $W$ with the assumpion used in \cite{Coleman86} and similar papers.

The existence of solitary waves in this situation is based on the fact that
the ratio between energy and charge can be sufficiently low; thus, following
\cite{bel}, these vortices are called hylomorphic (cf. section \ref{hc}).

Moreover these vortices are related with a non trivial magnetic field and a
non trivial electric field. In particular the magnetic field looks like to the
field created by a finite solenoid.

Since the KGM are invariant for the Lorentz group, a Lorentz boost of a vortex
creates a travelling solitary wave.

The paper is organized as follows. In section 2 we introduce the
KGM-equations, we study some of their general features, we give the definition
of three dimensional vortex and finally state the main result in Theorem
\ref{main}. Section 3 is devoted to the proof of Theorem \ref{main}.

\section{Statement of the problem and results}

\subsection{The Klein-Gordon-Maxwell system}

The nonlinear Klein-Gordon equation\ for a complex valued field $\psi,$
defined on the spacetime $\mathbb{R}^{4},$ can be written as follows:%
\begin{equation}
\square\psi+W^{\prime}(\left\vert \psi\right\vert )\frac{\psi}{\left\vert
\psi\right\vert }=0\label{KG}%
\end{equation}
where
\[
\square\psi=\frac{\partial^{2}\psi}{\partial t^{2}}-\Delta\psi,\;\;\text{\ }%
\Delta\psi=\frac{\partial^{2}\psi}{\partial x_{1}^{2}}+\frac{\partial^{2}\psi
}{\partial x_{2}^{2}}+\frac{\partial^{2}\psi}{\partial x_{3}^{2}}%
\]
and $W:\mathbb{R}_{+}\rightarrow\mathbb{R.}$

Hereafter $x=(x_{1},x_{2},x_{3})$ and $t$ will denote the space and time variables.

The field $\psi:$ $\mathbb{R}^{4}\rightarrow\mathbb{C}$ will be called
\textit{matter field}. If $W^{\prime}(s)$ is linear, $W^{\prime}(s)=m_{0}%
^{2}s,$ $m_{0}\neq0,$ equation (\ref{KG}) reduces to the Klein-Gordon equation.

Now let $\Gamma$ be a 1- form on $\mathbb{R}^{4}$ whose coefficients
$\Gamma_{j}$ are in the Lie algebra $u(1)$ of the group $U(1)=S^{1},$ i.e.
$\Gamma_{j}=-iA_{j}$, where $i$ is the imaginary unit and $A_{j}$ ($j=0,..,3)
$ are real maps defined in $\mathbb{R}^{4}.$

Consider the Abelian gauge theory related to $\psi$ and to $\Gamma$ and
described by the Lagrangian density (see e.g. \cite{yangL}, \cite{rub})
\begin{equation}
\mathcal{L}=\mathcal{L}_{0}+\mathcal{L}_{1}-W(\left\vert \psi\right\vert
)\label{marisa}%
\end{equation}
where
\[
\mathcal{L}_{0}=-\frac{1}{2}\left\langle d_{A}\psi,d_{A}\psi\right\rangle
,\text{ }\mathcal{L}_{1}=-\frac{1}{2}\left\langle d_{A}A,d_{A}A\right\rangle
,\text{ }A=\overset{3}{\underset{j=0}{\sum}}A_{j}dx^{j}%
\]
and $W$ is a real $C^{1}$-function. Here
\[
d_{A}=d-iqA=\overset{3}{\underset{j=0}{\sum}}\left(  \frac{\partial}{\partial
x^{j}}-iqA_{j}\right)
\]
denotes the gauge covariant differential and $\left\langle \cdot
,\cdot\right\rangle $ denotes the scalar product between forms with respect
the Minkowski metric in $\mathbb{R}^{4}$ and $q$ is a constant.

Since the $A_{j}$'s are real,
\[
d_{A}A=dA-iA\wedge A=dA.
\]

From now on, we shall use the following notation:%
\[
\mathbf{%
A\mathbf{=(}%
}A_{1},A_{2},A_{3}\mathbf{%
)%
}\text{ and }\phi=-A_{0}.
\]
If we set $t=x^{0}=-x_{0}$ and $x=(x^{1},x^{2},x^{3})=(x_{1},x_{2},x_{3})$ and
use vector notation, the Lagrangian densities $\mathcal{L}_{0},\mathcal{L}%
_{1}$ can be written as follows
\begin{equation}
\mathcal{L}_{0}=\frac{1}{2}\left[  \left\vert \left(  \partial_{t}%
+iq\phi\right)  \psi\right\vert ^{2}-\left\vert \left(  \nabla-iq\mathbf{A}%
\right)  \psi\right\vert ^{2}\right]  .
\end{equation}

\[
\mathcal{L}_{1}=\frac{1}{2}\left\vert \partial_{t}\mathbf{%
A%
}+\nabla\phi\right\vert ^{2}-\frac{1}{2}\left\vert \nabla\times\mathbf{A}%
\right\vert ^{2}.
\]
Here $\nabla\times$ and $\nabla$ denote respectively the curl and the gradient
operators with respect to the $x$ variable.

Now consider the total action of the Abelian gauge theory
\begin{equation}
\mathcal{S}=\int\left(  \mathcal{L}_{0}+\mathcal{L}_{1}-W(\left\vert
\psi\right\vert )\right)  dxdt.
\end{equation}

Making the variation of $\mathcal{S}$ with respect to $\psi,$ $\phi$ and
$\mathbf{A}$ we get the system of equations (KGM)
\begin{equation}
\left(  \partial_{t}+iq\phi\right)  ^{2}\psi-\left(  \nabla-iq\mathbf{A}%
\right)  ^{2}\psi+W^{^{\prime}}(\left\vert \psi\right\vert )\frac{\psi
}{\left\vert \psi\right\vert }=0\label{e1}%
\end{equation}%
\begin{equation}
\nabla\cdot\left(  \partial_{t}\mathbf{%
A%
}+\nabla\phi\right)  =q\left(  \operatorname{Im}\frac{\partial_{t}\psi}{\psi
}+q\phi\right)  \left\vert \psi\right\vert ^{2}\label{e2}%
\end{equation}%
\begin{equation}
\nabla\times\left(  \nabla\times\mathbf{A}\right)  +\partial_{t}\left(
\partial_{t}\mathbf{%
A%
}+\nabla\phi\right)  =q\left(  \operatorname{Im}\frac{\nabla\psi}{\psi
}-q\mathbf{A}\right)  \left\vert \psi\right\vert ^{2}.\label{e3}%
\end{equation}
Here $\nabla\cdot$ denotes the divergence operator.

In order to show the relation of the above equations with the Maxwell
equations and to get a model for Electrodynamics, we make the following change
of variables:
\begin{equation}
\mathbf{E=-}\left(  \frac{\partial\mathbf{%
A%
}}{\partial t}+\nabla\phi\right) \label{pos1}%
\end{equation}%
\begin{equation}
\mathbf{H}=\nabla\times\mathbf{A}\label{pos2}%
\end{equation}%
\begin{equation}
\rho=-q\left(  \operatorname{Im}\frac{\partial_{t}\psi}{\psi}+q\phi\right)
\left\vert \psi\right\vert ^{2}\label{caricona}%
\end{equation}%
\begin{equation}
\mathbf{j}=q\left(  \operatorname{Im}\frac{\nabla\psi}{\psi}-q\mathbf{A}%
\right)  \left\vert \psi\right\vert ^{2}.\label{tar}%
\end{equation}
So (\ref{e2}) and (\ref{e3}) are the second couple of the Maxwell equations
with respect to a matter distribution whose electric charge and current
densities are respectively $\rho$ and $\mathbf{j}$:
\begin{equation}
\nabla\cdot\mathbf{E}=\rho\label{gauss}%
\end{equation}%
\begin{equation}
\nabla\times\mathbf{H}-\frac{\partial\mathbf{E}}{\partial t}=\mathbf{j}%
.\label{ampere}%
\end{equation}
Equations (\ref{pos1}) and (\ref{pos2}) give rise to the first couple of the
Maxwell equation:
\begin{equation}
\nabla\times\mathbf{E}+\frac{\partial\mathbf{H}}{\partial t}=0\label{faraday}%
\end{equation}%
\begin{equation}
\nabla\cdot\mathbf{H}=0.\label{monopole}%
\end{equation}

If we set%

\[
\psi\left(  t,x\right)  =u\left(  t,x\right)  e^{iS(t,x)},\ u\in\mathbb{R}%
^{+},\ S\in\frac{\mathbb{R}}{2\pi\mathbb{Z}}%
\]
equation (\ref{e1}) can be splitted in the two following ones
\begin{equation}
\square u+W^{\prime}(u)+\left[  \left\vert \nabla S-q\mathbf{A}\right\vert
^{2}-\left(  \frac{\partial S}{\partial t}+q\phi\right)  ^{2}\right]  \,u=0
\end{equation}%
\begin{equation}
\frac{\partial}{\partial t}\left[  \left(  \frac{\partial S}{\partial t}%
+q\phi\right)  u^{2}\right]  -\nabla\cdot\left[  \left(  \nabla S-q\mathbf{A}%
\right)  u^{2}\right]  =0
\end{equation}
and these equations, using the varaiables $\mathbf{j}$ and $\rho$ become
\begin{equation}
\square u+W^{\prime}(u)+\frac{\mathbf{j}^{2}-\rho^{2}}{q^{2}u^{3}%
}=0\label{materia}%
\end{equation}%
\begin{equation}
\frac{\partial\rho}{\partial t}+\nabla\cdot\mathbf{j}=0.\label{continuità}%
\end{equation}
Equation (\ref{continuità}) is the charge continuity equation.

Notice that equation (\ref{continuità}) is a consequence of (\ref{gauss}) and
(\ref{ampere}).

In conclusion, an Abelian gauge theory, via equations (\ref{materia}%
,\ref{gauss},\ref{ampere},\ref{faraday},\ref{monopole}), provides a model of
interaction of the matter field $\psi$ with the electromagnetic field
$(\mathbf{E},\mathbf{H})$.

\subsection{The Hamilton-Jacobi equations\label{hc}}

Observe that the Lagrangian (\ref{marisa}) is invariant with respect to the
gauge transformations%
\begin{equation}
\psi\rightarrow e^{iq\chi}\psi\label{ga}%
\end{equation}%
\begin{equation}
\phi\rightarrow\phi-\partial_{t}\chi\label{ge}%
\end{equation}%
\begin{equation}
\mathbf{A\rightarrow A}+\nabla\chi\label{gi}%
\end{equation}
where $\chi\in C^{\infty}\left(  \mathbb{R}^{4}\right)  $.

So, our equations are gauge invariant; if we use the variable $u,\rho
,\mathbf{j},\mathbf{E},$ $\mathbf{H,}$ this fact can be checked directly since
these variables are gauge invariant.

In fact, equations (\ref{gauss},\ref{ampere},\ref{faraday},\ref{monopole}%
,\ref{materia}) are the gauge invariant formulation of equations
(\ref{e1},\ref{e2},\ref{e3}).

Also, we can replace the variables $\rho,\mathbf{j}$ with the variables
$\Omega$ and $\mathbf{K}$ defined by the following equations:%
\begin{align}
\rho & =q\Omega u^{2}\label{ch}\\
\mathbf{j}  & =q\mathbf{K}u^{2}.\label{cur}%
\end{align}

Using this notation, the continuity equation (\ref{continuità}) becomes%
\[
\partial_{t}(\Omega u^{2})+\nabla\cdot\left(  \mathbf{K}u^{2}\right)  =0.
\]
This equation allows us to interprete the matter field to be a fluid composed
by particles whose density is given by $\Omega u^{2}$ and which move in a
velocity field%
\begin{equation}
\mathbf{v}=\frac{\mathbf{K}}{\Omega}=-\frac{\nabla S-q\mathbf{A}}{\partial
_{t}S+q\phi}.\label{vel}%
\end{equation}
Then, since $\rho$ represents the electric charge density, $q=\rho/\Omega
u^{2}$ is interpreted as the electric charge of each particle. The total
number of particles
\begin{equation}
\sigma=\int\Omega u^{2}dx=-\int(\partial_{t}S+q\phi)u^{2}dx\label{apr}%
\end{equation}
is an integral of motion which, following \cite{bel}, we will call
\textit{hylenic charge}.

We set%
\begin{equation}
W(s)=\frac{m^{2}}{2}s^{2}+N(s),\label{enne}%
\end{equation}
with $N(0)=N^{\prime}(0)=N^{\prime\prime}(0)=0;$ then Equation (\ref{materia}%
), using the variables $\Omega$ and $\mathbf{K,}$ becomes%
\begin{equation}
\square u+N^{\prime}(u)+(m^{2}+\mathbf{K}^{2}-\Omega^{2})u=0\label{ganzetta}%
\end{equation}

If
\begin{equation}
\square u+N^{\prime}(u)\ll u,\label{rosa}%
\end{equation}
this equation, can be approximated by%
\[
\Omega^{2}=m^{2}+\mathbf{K}^{2};
\]
and, using the definition of $\Omega$ and $\mathbf{K}$ we get%
\[
\left(  \partial_{t}S+q\phi\right)  ^{2}=m^{2}+\left(  \nabla S-q\mathbf{A}%
\right)  ^{2}%
\]
or%
\begin{equation}
\partial_{t}S=-q\phi+\sqrt{m^{2}+\left(  \nabla S-q\mathbf{A}\right)  ^{2}%
}.\label{hj}%
\end{equation}
This is the relativistic Hamilton-Jacobi equation of a particle of rest mass
$m$ and charge $q$ in a elecromagnetic field with gauge potentials $\left(
\phi,\mathbf{A}\right)  $ (cf. e.g. \cite{landau} Ch. III). Equations
(\ref{hj}) and (\ref{vel}) completely describe the motion of these particles.

Since $S$ is a phase, then
\begin{align*}
\partial_{t}S  & =\omega\\
\nabla S  & =\mathbf{k};
\end{align*}
where $\omega$ and $\mathbf{k}$ are the local frequency and the local wave
number repectively. Moreover, the energy of each particle moving according to
(\ref{hj}), is given by%
\[
E=\partial_{t}S
\]
and its momentum is given by%
\[
\mathbf{p}=\nabla S;
\]
thus we have that%
\begin{align*}
E  & =\omega\\
\mathbf{p}  & =\mathbf{k};
\end{align*}
these two equations are the De Broglie relation with $\hbar=1;$ it is
interesting to see how the De Broglie relations arise in a natural way also
out of quantum mechanics. We notice that in the De Broglie interpretation of
quantum mechanics, $\Omega$ and $\mathbf{K}$ represent the probability density
and the probability flow of the position of a particle (see \cite{DB}).

If we do not assume (\ref{rosa}), equation (\ref{hj}) needs to be replaced by%
\begin{equation}
\partial_{t}S=-q\phi+\sqrt{m^{2}+\left(  \nabla S-q\mathbf{A}\right)
^{2}+\frac{\square u+N^{\prime}(u)}{u}}.\label{hjq}%
\end{equation}
Concluding, we may think that equation (\ref{e1}) describes a fluid of
particles of mass $m$ and charge $q$ which moves under the action of an
electromagnetic field $(\mathbf{E,H});$ the term $\frac{\square u+N^{\prime
}(u)}{u}$ in (\ref{hjq}) can be regarded as a field describing a sort of
interaction between particles. In the Bohm-De Broglie formulation of quantum
mechanics, this term corresponds to the \textit{quantum potential.}

\subsection{Conservation laws}

Noether's theorem states that any invariance for a one-parameter group of the
Lagrangian implies the existence of an integral of motion (see e.g.
\cite{gelfand}).

In the previuos section, we have seen that the hylenic charge $\sigma$ and,
consequently, the electric charge $Q=q\sigma$ are integrals of motions. This
conservation law is due to the gauge invariance.

Now we will consider other integrals which will be relevant for this paper.

\begin{itemize}
\item \textbf{Energy}. Energy, by definition, is the quantity which is
preserved by the time invariance of the Lagrangian; using the gauge invariant
variables, it takes the following form%
\begin{equation}
\mathcal{E}=\mathcal{E}_{m}+\mathcal{E}_{f}\label{sp}%
\end{equation}

where
\[
\mathcal{E}_{m}=\frac{1}{2}\int\left[  \left(  \frac{\partial u}{\partial
t}\right)  ^{2}+\left\vert \nabla u\right\vert ^{2}+(m^{2}+\Omega
^{2}+\mathbf{K}^{2})u^{2}\right]  +\int N(u)
\]
and%
\[
\mathcal{E}_{f}=\frac{1}{2}\int\left(  \mathbf{E}^{2}+\mathbf{H}^{2}\right)
dx.
\]
(for the computation of $\mathcal{E}$, see e.g. (\cite{befogranas})).

\item \textbf{Momentum. }Momentum, by definition, is the quantity which is
preserved by the space invariance of the Lagrangian; using the gauge invariant
variables, it takes the following form%
\begin{equation}
\mathbf{P}=\mathbf{P}_{m}+\mathbf{P}_{f}\label{spl}%
\end{equation}

where
\[
\mathbf{P}_{m}=\int\left[  -\left(  \partial_{t}u\,\nabla udx\right)
+\mathbf{K}\Omega u^{2}\right]  dx
\]

\end{itemize}

and%
\[
\mathbf{P}_{f}=\int\mathbf{E}\times\mathbf{H}\ dx.
\]

\begin{itemize}
\item \textbf{Angular momentum. }The angular momentum, by definition, is the
quantity which is preserved by virtue of the invariance under space rotations
of the Lagrangian with respect to the origin. Using the gauge invariant
variables, we get:%
\begin{equation}
\mathbf{M}=\mathbf{M}_{m}+\mathbf{M}_{f}\label{spli}%
\end{equation}
where
\begin{equation}
\mathbf{M}_{m}=\int\left[  -\mathbf{x}\times\left(  \nabla u\,\partial
_{t}u\right)  +\mathbf{x}\times\mathbf{K}\Omega u^{2}\right]  dx\label{am}%
\end{equation}
and%
\[
\mathbf{M}_{f}=\int\mathbf{x}\times\left(  \mathbf{E}\times\mathbf{H}\right)
\ dx.
\]

\end{itemize}

Notice that each of the integrals $\mathcal{E}$, $\mathbf{P},\mathbf{M}$ can
be splitted in two parts (see (\ref{sp}), (\ref{spl}), (\ref{spli})). The
first one refers to the "matter field" and the second to the "elecromagnetic field".

\subsection{Stationary solutions and vortices}

We look for stationary solutions of (\ref{e1}), (\ref{e2}), (\ref{e3}), namely
solutions of the form
\begin{align}
\psi\left(  t,x\right)   & =u\left(  x\right)  e^{iS(x,t)},\text{\ }%
u\in\mathbb{R}^{+},\ \omega\in\mathbb{R},\text{\ }S=S_{0}(x)-\omega t\in
\frac{\mathbb{R}}{2\pi\mathbb{Z}}\label{st}\\
\partial_{t}\mathbf{A}  & =0\mathbf{,\ }\partial_{t}\phi=0.\label{str}%
\end{align}

Substituting (\ref{st}) and (\ref{str}) in (\ref{e1}), (\ref{e2}), (\ref{e3}),
we get the following equations:
\begin{equation}
-\Delta u+\left[  \left\vert \nabla S_{0}-q\mathbf{A}\right\vert ^{2}-\left(
\omega-q\phi\right)  ^{2}\right]  \,u+W^{\prime}\left(  u\right)  =0\label{h1}%
\end{equation}%
\begin{equation}
-\nabla\cdot\left[  \left(  \nabla S_{0}-q\mathbf{A}\right)  u^{2}\right]
=0\label{h2}%
\end{equation}%
\begin{equation}
-\Delta\phi=q\left(  \omega-q\phi\right)  u^{2}\;\label{h3}%
\end{equation}%
\begin{equation}
\nabla\times\left(  \nabla\times\mathbf{A}\right)  =q\left(  \nabla
S_{0}-q\mathbf{A}\right)  u^{2}\;.\label{h4}%
\end{equation}
Observe that equation (\ref{h2}) easily follows from equation (\ref{h4}). Then
we are reduced to study the system (\ref{h1}), (\ref{h3}), (\ref{h4}). The
energy of a solution of equations (\ref{h1}), (\ref{h3}), (\ref{h4}) has the
following expression
\begin{align}
\mathcal{E}  & =\frac{1}{2}\int\left(  \left\vert \nabla u\right\vert
^{2}+\left\vert \nabla\phi\right\vert ^{2}+\left\vert \nabla\times
\mathbf{A}\right\vert ^{2}+(\left\vert \nabla S_{0}-q\mathbf{A}\right\vert
^{2}+\left(  \omega-q\phi\right)  ^{2})\,u^{2}\right) \nonumber\\
& +\int W(u)\label{enna}%
\end{align}

Moreover the (electric) charge (see (\ref{ch}) and (\ref{apr}) ) is given by
\begin{equation}
Q=q\sigma\label{ker}%
\end{equation}

where
\begin{equation}
\sigma=\int\Omega u^{2}=\int\left(  \omega-q\phi\right)  u^{2}.\label{car}%
\end{equation}

Clearly, when $u=0,$ the only finite energy gauge potentials which solve
(\ref{h3}), (\ref{h4}) are the trivial ones $\mathbf{A=}0,$ $\phi=0.$

It is possible to have three types of finite energy stationary non trivial solutions:

\begin{itemize}
\item electrostatic solutions: $\mathbf{A}=0$, $\phi\neq0;$

\item magnetostatic solutions: $\mathbf{A}\neq0$, $\phi=0;$

\item electro-magneto-static solutions: $\mathbf{A}\neq0$, $\phi\neq0$.
\end{itemize}

Under suitable assumptions, all these types of solutions exist. The existence
and the non existence of electrostatic solutions for the equations (\ref{h1}),
(\ref{h3}) have been proved under different assumptions on $W.$ In \cite{bf},
\cite{ca}, \cite{tea}, \cite{tea2}, \cite{dav} lower order terms $W$ like
(\ref{oc}) have been taken into account. In \cite{befo} the existence of
electrostatic solutions has been proved for a class of positive lower order
terms $W.$ In particular the existence of radially symmetric, electrostatic
solutions has been analyzed. These solutions have zero angular momentum.

Here we are interested in electro-magneto-static solutions, in particular we
shall study the existence of vortices which are solutions with nonvanishing
angular momentum.

We set
\[
\Sigma=\left\{  \left(  x_{1},x_{2},x_{3}\right)  \in\mathbb{R}^{3}%
:x_{1}=x_{2}=0\right\}
\]
and we define the map
\[
\theta:\mathbb{R}^{3}\backslash\Sigma\rightarrow\frac{\mathbb{R}}%
{2\pi\mathbb{Z}}%
\]%
\[
\theta(x_{1},x_{2},x_{3})=\operatorname{Im}\log(x_{1}+ix_{2}).
\]
In (\ref{st}) we take $S_{0}=\ell\theta$ ($\ell$ integer) and give the
following definition.

\begin{definition}
A finite energy solution of Eq. (\ref{h1}), (\ref{h3}), (\ref{h4}) is called
vortex if $S_{0}=\ell\theta(x)$ with $\ell\neq0.\ $
\end{definition}

In this case, $\psi$ has the following form
\begin{equation}
\psi(t,x)=u(x)\,e^{i\left(  \ell\theta(x)-\omega t\right)  };\ \ell
\in\mathbb{Z-}\left\{  0\right\}  .\label{ans}%
\end{equation}

We shall see (Proposition \ref{ang}) that the angular momentum $\mathbf{M}%
_{m}$ of the matter field of a vortex does not vanish; this fact justifies the
name "vortex".

Observe that $\theta\in C^{\infty}\left(  \mathbb{R}^{3}\backslash\Sigma
,\frac{\mathbb{R}}{2\pi\mathbb{Z}}\right)  .$ We set with abuse of notation%
\[
\nabla\theta(x)=\frac{x_{2}}{x_{1}^{2}+x_{2}^{2}}\mathbf{e}_{1}-\frac{x_{1}%
}{x_{1}^{2}+x_{2}^{2}}\mathbf{e}_{2}%
\]
where $\mathbf{e}_{1},\mathbf{e}_{2},\mathbf{e}_{3}$ is the canonical base in
$\mathbb{R}^{3}.$

Using the ansatz (\ref{ans}), equations (\ref{h1}), (\ref{h3}), (\ref{h4})
become
\begin{equation}
-\Delta u+\left[  \left\vert \ell\nabla\theta-q\mathbf{A}\right\vert
^{2}-\left(  \omega-q\phi\right)  ^{2}\right]  \,u+W^{\prime}(u)=0\label{z1}%
\end{equation}%
\begin{equation}
-\Delta\phi=q\left(  \omega-q\phi\right)  u^{2}\;\label{z3}%
\end{equation}%
\begin{equation}
\nabla\times\left(  \nabla\times\mathbf{A}\right)  =q\left(  \ell\nabla
\theta-q\mathbf{A}\right)  u^{2}.\;\label{z4}%
\end{equation}

In this case, the gauge invariant variables take the following expression:%
\begin{align}
\mathbf{E}  & =-\nabla\phi\label{1}\\
\mathbf{H}  & =\nabla\times\mathbf{A}\label{2}\\
\Omega & =\omega-q\phi\label{3}\\
\mathbf{K}  & =\ell\nabla\theta-q\mathbf{A}\label{4}\\
\rho & =q\Omega u^{2}\label{5}\\
\mathbf{j}  & =q\mathbf{K}u^{2}.\label{6}%
\end{align}
which give the equations%
\begin{align*}
\Delta u-W^{\prime}(u)  & =\left(  \mathbf{K}^{2}-\Omega^{2}\right)  u\\
\nabla\cdot\mathbf{E}  & =q\Omega u^{2}\\
\nabla\times\mathbf{H}  & =q\mathbf{K}u^{2}.
\end{align*}

As we have observed in the introduction, positive, double well shaped
potentials $W$ like (\ref{v1}) are not suitable for the existence of
3-dimensional vortices of type (\ref{ans}). In fact the following proposition holds:

\begin{proposition}
\label{b}Assume that $W$ satisfies the assumptions:%
\begin{equation}
\forall s\geq0:\ W(s)\geq0\label{1w}%
\end{equation}%
\begin{equation}
W(0)>0\label{2w}%
\end{equation}%
\begin{equation}
\exists\bar{s}:W(\bar{s})=0\label{3ww}%
\end{equation}
then (\ref{h1}), (\ref{h3}), (\ref{h4}) has no vortex solution..
\end{proposition}

\begin{proof}
We shall prove that any configuration of the type
\begin{equation}
(u\left(  x\right)  e^{i\left(  \ell\theta-\omega t\right)  },\phi
,\mathbf{A})\text{, }\ell\in\mathbb{Z},\ell\neq0\label{one}%
\end{equation}
has infinite energy $\mathcal{E}$ (\ref{enna}). Arguing by contradiction
assume that (\ref{one}) has finite energy. Since $W$ satisfies (\ref{1w})
(\ref{2w}) and \ref{3ww}), the finiteness of the energy implies
\[
\int W(u)<\infty
\]
so that
\[
u(\infty)=\bar{s}.
\]
So$,$ using again the finiteness of $\mathcal{E},$ we get
\[
\int\left\vert \ell\nabla\theta-q\mathbf{A}\right\vert ^{2}<\infty.
\]
So, if we take $0<\delta_{1}<\delta_{2},$ for all $\varepsilon>0$ there exists
$M>0$ s.t. for all $x=\left(  x_{1},x_{2,}x_{3}\right)  $ with
\[
\delta_{1}<r<\delta_{2},\text{ }\left\vert x_{3}\right\vert >M,\text{ }%
r=\sqrt{x_{1}^{2}+x_{2}^{2}\text{ }},
\]
we have
\[
\left\vert \ell\nabla\theta-q\mathbf{A}\right\vert <\varepsilon.
\]
So, for such $\left(  x_{1},x_{2,}x_{3}\right)  ,$ we get
\[
\frac{\left\vert \ell\right\vert }{\delta_{2}}-\varepsilon<\frac{\left\vert
\ell\right\vert }{r}-\varepsilon=\left\vert \ell\nabla\theta\right\vert
-\varepsilon<\left\vert \mathbf{A}\left(  x\right)  \right\vert .
\]
Then, if $\varepsilon$ is small enough, we get
\[
0<\mu=\frac{\left\vert \ell\right\vert }{\delta_{2}}-\varepsilon<\left\vert
\mathbf{A}\left(  x\right)  \right\vert .
\]
So
\[
\infty=\int_{\delta_{1}}^{\delta_{2}}rdr\int_{\text{ }\left\vert
x_{3}\right\vert >M}\mu^{6}dx_{3}\leq\int\left\vert \mathbf{A}\left(
x\right)  \right\vert ^{6}dx.
\]
Therefore $\mathbf{A}\notin L^{6}(\mathbb{R}^{3})$ and then, by Sobolev
inequality,
\[%
{\displaystyle\int}
\left\vert \nabla\mathbf{A}\right\vert ^{2}dx=\infty.
\]
This contradicts the finiteness of the energy $\mathcal{E}$.
\end{proof}

\subsection{The main existence result}

Let $W$ satisfy the following assumptions:

\begin{itemize}
\item W1) $\forall s\geq0:\ W(s)\geq0$

\item W2) $W$ is $C^{2}$ with $W(0)=W^{\prime}(0)=0$, $W^{\prime\prime
}(0)=m^{2}>0,$

\item W3) $\underset{s>0}{\inf}\left(  \frac{W(s)}{\frac{m^{2}}{2}s^{2}%
}\right)  <1$
\end{itemize}

We shall set
\[
W(s)=\frac{m^{2}}{2}s^{2}+N(s).
\]
Clearly assumption W3) is equivalent to require that there exists $s_{0}$ $>0
$ such that
\begin{equation}
N(s_{0})<0.\label{pec}%
\end{equation}

By rescaling time and space we can assume without loss of generality
\[
m^{2}=1.
\]
Moreover, for technical reasons it is useful to assume that $W$ is defined for
all $s\in\mathbb{R}$ just setting
\[
W(s)=W(-s)\ \ for\ s<0.
\]

Now we can state the main existence result.

\begin{theorem}
\label{main}Assume that the function $W$ satisfies assumptions W1),W2),W3).
Then for all $\ell\in\mathbb{Z}$ there exists $\bar{q}>0$ such that for every
$0\leq q\leq\bar{q}$ \ the equations (\ref{z1}), (\ref{z3}), (\ref{z4}) admit
a finite energy solution in the sense of distributions $(u,\omega
,\phi\mathbf{,A}),$ $u\neq0,$ $\omega>0.$ The maps $u,$ $\phi$ depend only on
the variables $r=\sqrt{x_{1}^{2}+x_{2}^{2}}$ and $x_{3}$
\[
u=u(r,x_{3}),\text{ }\phi=\phi(r,x_{3}).
\]
and the magnetic potential $\mathbf{A}$ has the following form%
\begin{equation}
\mathbf{A}=a(r,x_{3})\nabla\theta=a(r,x_{3})\left(  \frac{x_{2}}{r^{2}%
}\mathbf{e}_{1}-\frac{x_{1}}{r^{2}}\mathbf{e}_{2}\right) \label{a}%
\end{equation}
If $q=0$, then $\phi=0,\mathbf{A}=0.$ If $q>0$ then $\phi$ $\neq0.$ Moreover
$\mathbf{A}$ $\neq0$ if and only if $\ell\neq0$
\end{theorem}

\begin{remark}
When there is no coupling with the electromagnetic field, i.e. $q=0,$
equations (\ref{z1}), (\ref{z3}), (\ref{z4}) reduce to find vortices to the
nonlinear Klein-Gordon equation and an analogous result has been obtained in
\cite{bad}.
\end{remark}

\begin{remark}
When $\ell=0$ and $q>0$ the last part of Theorem \ref{main} states the
existence of electrostatic solutions, namely finite energy solutions with
$u\neq0,$ $\phi$ $\neq0$ and $\mathbf{A}$ $=0.$This result generalizes a
recent theorem (see \cite{befo}), where the existence of electrostatic
solutions has been stated under assumptions stronger than W1),W2),W3).
\end{remark}

\begin{remark}
\label{v} By the presence of the term $\nabla\theta$ equations (\ref{z1}),
(\ref{z4}) are not invariant under the $O(3)$ group action as it happens for
the equations (\ref{e1}), (\ref{e2}), (\ref{e3}) we started from. Indeed there
is a breaking of radial symmetry and the solutions $u$, $\phi,$ $\mathbf{A}$
in theorem \ref{main} have only an $S^{1}$ (cylindrical) symmetry.
\end{remark}

\begin{proposition}
\label{ang}Let $(u,\omega,\phi\mathbf{,A})$ be a non trivial, finite energy
solution of equations (\ref{z1}), (\ref{z3}), (\ref{z4}) as in theorem
\ref{main}. Then the angular momentum $\mathbf{M}_{m}$ (see (\ref{am})) has
the following expression
\begin{equation}
\mathbf{M}_{m}=-\left[  \int\left(  \ell-qa\right)  \left(  \omega
-q\phi\right)  u^{2}dx\right]  \mathbf{e}_{3}\label{vc}%
\end{equation}
and, if $\ell\neq0,$ it does not vanish.
\end{proposition}

\begin{proof}
By (\ref{4}) and (\ref{a}), we have that%
\begin{equation}
\mathbf{K}=\nabla S-q\mathbf{A}=\ell\nabla\theta-qa\nabla\theta=\left(
\ell-qa\right)  \nabla\theta.\label{new}%
\end{equation}

Then, using (\ref{3}) and (\ref{new}), we have that%
\[
\mathbf{M}_{m}=\int\mathbf{x}\times\mathbf{K}\Omega u^{2}dx=\int
\mathbf{x}\times\nabla\theta\left(  \ell-qa\right)  \left(  \omega
-q\phi\right)  u^{2}dx.
\]
Let us compute%
\begin{align*}
\mathbf{x}\times\nabla\theta & =(x_{1}\mathbf{e}_{1}+x_{2}\mathbf{e}_{2}%
+x_{3}\mathbf{e}_{3})\times\left(  \frac{x_{2}}{r^{2}}\mathbf{e}_{1}%
-\frac{x_{1}}{r^{2}}\mathbf{e}_{2}\right) \\
& =-\frac{x_{1}^{2}}{r^{2}}\mathbf{e}_{3}-\frac{x_{2}^{2}}{r^{2}}%
\mathbf{e}_{3}+\frac{x_{2}x_{3}}{r^{2}}\mathbf{e}_{2}+\frac{x_{1}x_{3}}{r^{2}%
}\mathbf{e}_{1}\\
& =\frac{x_{1}x_{3}}{r^{2}}\mathbf{e}_{1}+\frac{x_{2}x_{3}}{r^{2}}%
\mathbf{e}_{2}-\mathbf{e}_{3}.
\end{align*}
Then%
\begin{equation}
\mathbf{M}_{m}\left(  \psi\right)  =\int\left(  \frac{x_{1}x_{3}}{r^{2}%
}\mathbf{e}_{1}+\frac{x_{2}x_{3}}{r^{2}}\mathbf{e}_{2}-\mathbf{e}_{3}\right)
\left(  \ell-qa\right)  \left(  \omega-q\phi\right)  u^{2}dx.\label{very}%
\end{equation}

On the other hand, since the functions $x_{1}x_{3}\frac{\left(  \ell
-qa\right)  \left(  \omega-q\phi\right)  u^{2}}{r^{2}}$ and $x_{2}x_{3}%
\frac{\left(  \ell-qa\right)  \left(  \omega-q\phi\right)  u^{2}}{r^{2}}$ are
odd in $x_{1}$ and $x_{2}$ respectively, we have
\begin{equation}
\int x_{1}x_{3}\frac{\left(  \ell-qa\right)  \left(  \omega-q\phi\right)
u^{2}}{r^{2}}\mathbf{=}\int x_{2}x_{3}\frac{\left(  \ell-qa\right)  \left(
\omega-q\phi\right)  u^{2}}{r^{2}}=0.\label{zero}%
\end{equation}

Then (\ref{vc}) follows from (\ref{very}) and (\ref{zero}). Now let $\ell
\neq0.$ In order to see that $\mathbf{M}_{m}\neq0,$ it is sufficient to prove
that
\begin{equation}
\left(  \ell-qa\right)  \left(  \omega-q\phi\right)  >0.\label{pat}%
\end{equation}
or that
\begin{equation}
\left(  \ell-qa\right)  \left(  \omega-q\phi\right)  <0.\label{pit}%
\end{equation}

Clearly, since $\ell,\omega\neq0$ (\ref{pat}) or (\ref{pit}) are satisfied
when $q=0.$ Now let $q>0.$ Assume that $\ell>0$ and we show that (\ref{pat})
is verified. The case $\ell<0$ can be treated analogously.

By (\ref{h3}) we have that%
\[
-\Delta\phi+q^{2}u^{2}\phi=q\omega u^{2}.
\]
Since $\omega/q$ is a supersolution, by the maximum principle, $\phi<\omega/q$
and hence $\omega-q\phi>0.$ So, in order to prove (\ref{pat}), it remains to
show that%
\begin{equation}
\ell-qa>0\label{rem}%
\end{equation}

By (\ref{h4}) we have that%
\begin{equation}
\nabla\times\left(  \nabla\times\mathbf{A}\right)  =q\left(  \ell\nabla
\theta-q\mathbf{A}\right)  u^{2}.\;\label{beco}%
\end{equation}

Now a straight computation shows that,
\begin{equation}
\nabla\times(\nabla\times a\nabla\theta)=b\ \nabla\theta\label{bricco}%
\end{equation}

where
\[
b=-\frac{\partial^{2}a}{\partial r^{2}}+\frac{1}{2}\frac{\partial a}{\partial
r}-\frac{\partial^{2}a}{\partial x_{3}^{2}}.
\]
Then, setting $\mathbf{A}=a\nabla\theta$ in (\ref{beco}) and using
(\ref{bricco}), we have%
\[
-\frac{\partial^{2}a}{\partial r^{2}}+\frac{1}{2}\frac{\partial a}{\partial
r}-\frac{\partial^{2}a}{\partial x_{3}^{2}}=q\left(  \ell-qa\right)  u^{2}.\;
\]
Since $\ell/q$ is a supersolution, by the maximum principle, $a<\ell/q$ and
hence (\ref{rem}) is proved.
\end{proof}

\begin{remark}
Observe that in the interpretation given in \ref{hc}, the quantity $\Omega
u^{2}=\left(  \omega-q\phi\right)  u^{2}$ represents the density of
particles$;\ $then by (\ref{vc}) $-\left(  \ell-qa\right)  \mathbf{e}_{3}$
represents the angular momentum of each particle. So, since $\ell$ is an
integer, we see that the classical model described by this Abelian gauge
theory presents a quantization phenomenon. Notice also that, for $q=0,$ the
angular momentum of each particle takes only integer values.
\end{remark}

Finally let us observe that under general assumptions on $W,$ magnetostatic
solutions (i.e. with $\omega=\phi=0)$ do not exist$.$ In fact the following
proposition holds:

\begin{proposition}
\label{A}Assume that $W$ satisfies the assumptions $W(0)=0$ and $W^{\prime
}(s)s\geq0.$ Then (\ref{z1}), (\ref{z3}), (\ref{z4}) has no solutions with
$\omega=\phi=0$.
\end{proposition}

\begin{proof}
Set $\omega=0,$ $\phi=0$ in (\ref{z1}) and we get
\[
-\Delta u+\left\vert \ell\nabla\theta-q\mathbf{A}\right\vert ^{2}%
\,u+W^{\prime}(u)=0.
\]
Then, multiplying by $u$ and integrating, we get
\[
\int\left\vert \nabla u\right\vert ^{2}+\left\vert \ell\nabla\theta
-q\mathbf{A}\right\vert ^{2}\,u^{2}+W^{\prime}(u)u=0.
\]
So, since $W^{\prime}(s)s\geq0,$ we get $u=0.$
\end{proof}

\section{The existence proof}

\subsection{The functional setting}

Let $H^{1}$ denote the usual Sobolev space with norm
\[
\left\Vert u\right\Vert _{H^{1}}^{2}=\int(\left\vert \nabla u\right\vert
^{2}+u^{2})dx;
\]
moreover we need to use also the weighted Sobolev space $\hat{H}^{1}$ whose
norm is given by
\[
\left\Vert u\right\Vert _{\hat{H}^{1}}^{2}=\int\left[  \left\vert \nabla
u\right\vert ^{2}+\left(  1+\frac{\ell^{2}}{r^{2}}\right)  u^{2}\right]
dx,\text{ }\ell\in\mathbb{Z}%
\]
where $r=\sqrt{x_{1}^{2}+x_{2}^{2}}.$ Clearly $\hat{H}^{1}=H^{1}$ when
$\ell=0.$

We set $\mathcal{D}=C_{0}^{\infty}(\mathbb{R}^{3})$ and we denote by
$\mathcal{D}^{1,2}$ the completion of $\mathcal{D}$ with respect to the inner
product
\begin{equation}
\left(  v\mid w\right)  _{\mathcal{D}^{1,2}}=\int\nabla v\cdot\nabla
wdx.\label{inner}%
\end{equation}
Here and in the following the dot $\cdot$ will denote the Euclidean inner
product in $\mathbb{R}^{3}.$

We set
\[
H=\hat{H}^{1}\times\mathcal{D}^{1,2}\times\left(  \mathcal{D}^{1,2}\right)
^{3}%
\]%
\begin{equation}
\left\Vert \left(  u,\phi,\mathbf{A}\right)  \right\Vert _{H}^{2}%
=\int\left\vert \nabla u\right\vert ^{2}+\left(  1+\frac{\ell^{2}}{r^{2}%
}\right)  u^{2}+\left\vert \nabla\phi\right\vert ^{2}+\left\vert
\nabla\mathbf{A}\right\vert ^{2}.\label{usual}%
\end{equation}
We shall denote by $u=u(r,x_{3})$ the real maps having cylindrical symmetry,
i.e. those real maps in $\mathbb{R}^{3}$ which depend only from $r=\sqrt
{x_{1}^{2}+x_{2}^{2}}$ and $x_{3}.$ We set%
\begin{equation}
\mathcal{D}_{r}=\left\{  u\in\mathcal{D}:u=u(r,x_{3})\right\} \label{cf}%
\end{equation}
and we shall denote by $\mathcal{D}_{r}^{1,2}$ (respectively $\hat{H}_{r}%
^{1})$ the closure of $\mathcal{D}_{r}$ in the $\mathcal{D}^{1,2}$
(respectively $\hat{H}^{1})$ norm.

Now we consider the functional
\begin{align}
J(u,\phi\mathbf{,A})  & =\frac{1}{2}\int\left\vert \nabla u\right\vert
^{2}-\left\vert \nabla\phi\right\vert ^{2}+\left\vert \nabla\times
\mathbf{A}\right\vert ^{2}\nonumber\\
& +\frac{1}{2}\int\left[  \left\vert \ell\nabla\theta-q\mathbf{A}\right\vert
^{2}-\left(  \omega-q\phi\right)  ^{2}\right]  \,u^{2}+\int
W(u)\label{functional}%
\end{align}

\noindent where $\left(  u,\phi,\mathbf{A}\right)  \in H.$ The equations
(\ref{z1}),\ (\ref{z3}) and (\ref{z4}) are the Euler-Lagrange equations of the
functional $J$. Standard computations show that the following lemma holds:

\ 

\begin{lemma}
\label{tecnico} Assume that $W$ satisfies W1), W2), W3) and%
\begin{equation}
\lim\sup\frac{\left.  W^{\prime}(s)\right.  }{s^{5}}<\infty\text{ for
}s\rightarrow\infty.\label{gratis}%
\end{equation}
Then the functional $J$ is $C^{1}$ on $H$.
\end{lemma}

Without loss of generality we can assume that assumption (\ref{gratis}) is
satisfied. In fact, if (\ref{gratis}) is not satisfied, we can replace
$W^{\prime}(s)$ with $W^{\prime}(\bar{s})s$ for $s>\bar{s}$, where $\bar{s}>0$
is s.t. $W^{\prime}(\bar{s})>0.$ By using a maximum principle argument, it can
be seen that, with this nonlinearity, any solution $u$ takes values between
$0$ and $\bar{s}.$

By the above lemma it follows that the critical points $\left(  u,\phi
,\mathbf{A}\right)  \in H$ of $J$ (with $u\geq0$) are weak solutions of eq.
(\ref{z1}),\ (\ref{z3}) and (\ref{z4}), namely
\begin{equation}
\int\nabla u\cdot\nabla v+\left[  \left\vert \ell\nabla\theta-q\mathbf{A}%
\right\vert ^{2}-\left(  \omega-q\phi\right)  ^{2}\right]  \,uv+W^{\prime
}\left(  u\right)  v=0,\;\forall v\in\hat{H}^{1}\label{w1}%
\end{equation}%
\begin{equation}
\int\nabla\phi\cdot\nabla w+qu^{2}\left(  \omega-q\phi\right)  w=0,\;\forall
w\in\mathcal{D}^{1,2}\label{we}%
\end{equation}
\begin{equation}
\int\nabla\mathbf{A}\cdot\nabla\mathbf{V}-qu^{2}\left(  \ell\nabla
\theta-q\mathbf{A}\right)  \cdot\mathbf{V}=0,\;\forall\mathbf{V}%
\in(\mathcal{D}^{1,2})^{3}.\label{w3}%
\end{equation}

\subsection{Solutions in the sense of distributions}

Since $\mathcal{D}$ is not contained in $\hat{H}^{1},$ a solution $\left(
u,\phi,\mathbf{A}\right)  \in H$ of \ (\ref{w1}), (\ref{we}), (\ref{w3}) need
not to be a solution of (\ref{z1}), (\ref{z3}), (\ref{z4}) in the sense of
distributions on $\mathbb{R}^{3}$. In fact, since $\nabla\theta\left(
x\right)  $ is singular on $\Sigma,$ it might be that for some test function
$v\in\mathcal{D},$ when $\ell\neq0,$ the integral $\int\left\vert \ell
\nabla\theta-q\mathbf{A}\right\vert ^{2}\,uv$ diverges, unless $u$ is
sufficiently small as $x\rightarrow\Sigma.$

In this section we will show that this fact does not occur, namely the
singularity is removable in the sense of the following theorem:

\begin{theorem}
\label{finale}Let $(u_{0},\phi_{0}\mathbf{,A}_{0})$ $\in H,$ $u_{0}\geq0$ be a
solution of (\ref{w1}), (\ref{we}), (\ref{w3}) (i.e. a critical point of $J).$
Then $(u_{0},\phi_{0}\mathbf{,A}_{0})$ is a solution of equations
(\ref{z1}),\ (\ref{z3}) and (\ref{z4}) in the sense of distribution, namely
\begin{equation}
\int\nabla u_{0}\cdot\nabla v+\left[  \left\vert \ell\nabla\theta
-q\mathbf{A}_{0}\right\vert ^{2}-\left(  \omega-q\phi_{0}\right)  ^{2}\right]
\,u_{0}v+W^{\prime}\left(  u_{0}\right)  v=0,\;\forall v\in\mathcal{D}%
\label{im}%
\end{equation}%
\begin{equation}
\int\nabla\phi_{0}\cdot\nabla w-qu_{0}^{2}\left(  \omega-q\phi_{0}\right)
w=0,\;\forall w\in\mathcal{D}\label{dis}%
\end{equation}%
\begin{equation}
\int\nabla\mathbf{A}_{0}\cdot\nabla\mathbf{V}-qu_{0}^{2}\left(  \ell
\nabla\theta-q\mathbf{A}_{0}\right)  \cdot\mathbf{V}=0,\;\forall\mathbf{V}%
\in\mathcal{D}^{3}.\label{um}%
\end{equation}

\end{theorem}

Let $\chi_{n}$ ($n$ positive integer) be a family of smooth functions
depending only on $r=\sqrt{x_{1}^{2}+x_{2}^{2}}$ and $x_{3}$ and which satisfy
the following assumptions:

\begin{itemize}
\item $\chi_{n}\left(  r,x_{3}\right)  =1$ for $r\geq\frac{2}{n}$

\item $\chi_{n}\left(  r,x_{3}\right)  =0$ for $r\leq\frac{1}{n}$

\item $\left|  \chi_{n}\left(  r,x_{3}\right)  \right|  \leq1$

\item $\left|  \nabla\chi_{n}\left(  r,x_{3}\right)  \right|  \leq2n$

\item $\chi_{n+1}\left(  r,x_{3}\right)  \geq\chi_{n}\left(  r,x_{3}\right)  $
\end{itemize}

\begin{lemma}
\label{caccona}Let $\varphi$ be a function in $H^{1}\cap L^{\infty}$ with
bounded support and set $\varphi_{n}=\varphi\cdot\chi_{n}.$ Then, up to a
subsequence, we have that
\[
\varphi_{n}\rightarrow\varphi\text{ weakly in }H^{1}%
\]

\end{lemma}

\begin{proof}
Clearly $\varphi_{n}\rightarrow\varphi$ $a.e.$ Then, by standard arguments,
the conclusion holds if we show that $\left\{  \varphi_{n}\right\}  $ is
bounded in $H^{1}.$ Clearly $\left\{  \varphi_{n}\right\}  $ is bounded in
$L^{2}.$ Let us now prove that
\[
\left\{  \int\left\vert \nabla\varphi_{n}\right\vert ^{2}\right\}  \text{ is
bounded.}%
\]
We have
\begin{align*}
\int\left\vert \nabla\varphi_{n}\right\vert ^{2}  & \leq2\int\left\vert
\nabla\varphi\cdot\chi_{n}\right\vert ^{2}+\left\vert \varphi\cdot\nabla
\chi_{n}\right\vert ^{2}\\
& \leq2\int\left\vert \nabla\varphi\right\vert ^{2}+2\int_{\Gamma
_{\varepsilon}}\left\vert \varphi\cdot\nabla\chi_{n}\right\vert ^{2}%
\end{align*}
where
\[
\Gamma_{\varepsilon}=\left\{  x\in\mathbb{R}^{3}:\varphi\neq0\ and\ \left\vert
\nabla\chi_{n}\left(  r,z\right)  \right\vert \neq0\right\}  .
\]
By our construction, $\left\vert \Gamma_{\varepsilon}\right\vert \leq c/n^{2}$
where $c$ depends only on $\varphi.$ Thus
\begin{align*}
\int\left\vert \nabla\varphi_{n}\right\vert ^{2}  & \leq2\int\left\vert
\nabla\varphi\right\vert ^{2}+2\left\Vert \varphi\right\Vert _{L^{\infty}}%
^{2}\int_{\Gamma_{\varepsilon}}\left\vert \nabla\chi_{n}\right\vert ^{2}\\
& \leq2\int\left\vert \nabla\varphi\right\vert ^{2}+2\left\Vert \varphi
\right\Vert _{L^{\infty}}^{2}\cdot\left\vert \Gamma_{\varepsilon}\right\vert
\cdot\left\Vert \nabla\chi_{n}\right\Vert _{L^{\infty}}^{2}\\
& \leq2\int\left\vert \nabla\varphi\right\vert ^{2}+8c\left\Vert
\varphi\right\Vert _{L^{\infty}}^{2}.
\end{align*}
Thus $\varphi_{n}$ is bounded in $H^{1}$ and $\varphi_{n}\rightarrow\varphi$
weakly in $H^{1}.$

\bigskip
\end{proof}

\bigskip

\bigskip Now we are ready to prove Theorem \ref{finale}

\begin{proof}
Clearly (\ref{dis}) and (\ref{um}) immediately follow by (\ref{we}) and
(\ref{w3}). Let us prove (\ref{im}). The case $\ell=0$ is trivial. So assume
$\ell\neq0.$ We take any $v\in\mathcal{D}$ and set $\varphi_{n}=v^{+}\chi_{n}$
where $v^{+}=\frac{\left\vert v\right\vert +v}{2}.$ Then, taking $\varphi_{n}$
as test function in Eq. (\ref{w1}), we have
\begin{equation}
\int\nabla u_{0}\cdot\nabla\varphi_{n}+\left[  \left\vert q\mathbf{A}_{0}%
-\ell\nabla\theta\right\vert ^{2}-\left(  q\phi_{0}-\omega\right)
^{2}\right]  \,u_{0}\varphi_{n}+W^{\prime}\left(  u_{0}\right)  \varphi
_{n}=0\ \label{ecco}%
\end{equation}

Equation (\ref{ecco}) can be written as follows
\begin{equation}
A_{n}+B_{n}+C_{n}+D_{n}=0\label{simbol}%
\end{equation}

where
\begin{equation}
A_{n}=\int\nabla u_{0}\cdot\nabla\varphi_{n},\text{ }B_{n}=\int\left(
q^{2}\mathbf{A}_{0}^{2}u_{0}-\left(  q\phi_{0}-\omega\right)  ^{2}%
u_{0}+W^{\prime}\left(  u_{0}\right)  \right)  \varphi_{n}\label{def1}%
\end{equation}

\begin{equation}
C_{n}=-2\int q\mathbf{A}_{0}\cdot\ell\nabla\theta\ u_{0}\varphi_{n},\text{
}D_{n}=\int\left\vert \ell\nabla\theta\right\vert ^{2}u_{0}\varphi
_{n}.\label{def2}%
\end{equation}
By Lemma \ref{caccona}
\begin{equation}
\varphi_{n}\rightarrow v^{+}\text{ weakly in }H^{1}.\label{w}%
\end{equation}

Then we have
\begin{equation}
A_{n}\rightarrow\int\nabla u_{0}\cdot\nabla v^{+}.\label{pa}%
\end{equation}

Now
\[
\left(  q^{2}\mathbf{A}_{0}^{2}u_{0}-\left(  q\phi_{0}-\omega\right)
^{2}u_{0}+W^{\prime}\left(  u_{0}\right)  \right)  \in L^{6/5}=\left(
L^{6}\right)  ^{\prime}.
\]
Then, using again (\ref{w}) and by the embedding $H^{1}\subset L^{6}$, we
have
\begin{equation}
B_{n}\rightarrow\int\left(  q^{2}\mathbf{A}_{0}^{2}u_{0}-\left(  q\phi
_{0}-\omega\right)  ^{2}u_{0}+W^{\prime}\left(  u_{0}\right)  \right)
v^{+}<\infty.\label{bi}%
\end{equation}

Now we shall prove that
\begin{equation}
C_{n}\rightarrow-2\int q\mathbf{A}_{0}\cdot\ell\nabla\theta\ u_{0}v^{+}%
<\infty.\label{ba}%
\end{equation}
Set
\[
C=B_{R}\times\left[  -d,d\right]  ,\text{ }B_{R}=\left\{  (x_{1},x_{2}%
)\in\mathbb{R}^{2}:r^{2}=x_{1}^{2}+x_{2}^{2}<R\right\}
\]
where $d,$ $R>0$ are so large that the cylinder $C$ contains the support of
$v^{+}.$

Then
\begin{align}
\int\left(  \frac{\varphi_{n}}{r}\right)  ^{\frac{3}{2}}dx  & =\int_{C}\left(
\frac{v^{+}\chi_{n}}{r}\right)  ^{\frac{3}{2}}dx\\
& \leq c_{1}\int_{-d}^{d}\int_{0}^{R}\left(  \frac{1}{r}\right)  ^{\frac{3}%
{2}}rdrdx_{3}=M<\infty\label{pos}%
\end{align}

where $c_{1}=2\pi\sup\left(  v^{+}\right)  ^{\frac{3}{2}}.$ By (\ref{pos}) we
have
\begin{equation}
\int\left\vert \mathbf{A}_{0}\cdot\nabla\theta\ u_{0}\varphi_{n}\right\vert
dx\leq\left\Vert u_{0}\mathbf{A}_{0}\ \right\Vert _{L^{3}}\left\Vert
\frac{\varphi_{n}}{r}\right\Vert _{L^{\frac{3}{2}}}\leq\left\Vert
u_{0}\mathbf{A}_{0}\ \right\Vert _{L^{3}}M^{\frac{2}{3}}.\label{posse}%
\end{equation}

Now
\[
\left\vert \mathbf{A}_{0}\cdot\nabla\theta\ u_{0}\varphi_{n}\right\vert
\rightarrow\left\vert \mathbf{A}_{0}\cdot\nabla\theta\ u_{0}v^{+}\right\vert
\text{ a.e. in }\mathbb{R}^{3}%
\]
and the sequence $\left\{  \left\vert \mathbf{A}_{0}\cdot\nabla\theta
\ u_{0}\varphi_{n}\right\vert \right\}  $ is monotone $.$ Then, by the
monotone convergence theorem, we get
\begin{equation}
\int\left\vert \mathbf{A}_{0}\cdot\ell\nabla\theta\ u_{0}\varphi
_{n}\right\vert dx\rightarrow\int\left\vert \mathbf{A}_{0}\cdot\ell
\nabla\theta\ u_{0}v^{+}\right\vert dx.\label{from}%
\end{equation}

By (\ref{posse}) and (\ref{from}) we deduce that
\begin{equation}
\int\left\vert \mathbf{A}_{0}\cdot\ell\nabla\theta\ u_{0}v^{+}\right\vert
dx<\infty.\label{byy}%
\end{equation}

Then, since
\[
\left\vert \mathbf{A}_{0}\cdot\nabla\theta\ u_{0}\varphi_{n}\right\vert
\leq\left\vert \mathbf{A}_{0}\cdot\nabla\theta\ u_{0}v^{+}\right\vert \in
L^{1},
\]
by the dominated convergence Theorem, we get (\ref{ba}). Finally we prove
that
\begin{equation}
D_{n}\rightarrow\int\left\vert \ell\nabla\theta\right\vert ^{2}u_{0}%
v^{+}<\infty.\label{fin}%
\end{equation}
By (\ref{simbol}), (\ref{pa}), (\ref{bi}) and (\ref{ba}) we have that
\begin{equation}
D_{n}=\int\left\vert \ell\nabla\theta\right\vert ^{2}u_{0}\varphi_{n}\text{ is
bounded.}\label{bu}%
\end{equation}
Then the sequence $\left\vert \nabla\theta\right\vert ^{2}\ u_{0}\varphi_{n}$
is monotone and it converges a.e$.$ to $\left\vert \nabla\theta\right\vert
^{2}\ u_{0}v^{+}.$ Then, by the monotone convergence theorem, we get
\begin{equation}
\int\left\vert \ell\nabla\theta\right\vert ^{2}\ u_{0}\varphi_{n}%
dx\rightarrow\int\left\vert \ell\nabla\theta\right\vert ^{2}\ u_{0}%
v^{+}dx.\label{bubu}%
\end{equation}
By (\ref{bu}) and (\ref{bubu}) we get (\ref{fin}).

Taking the limit in (\ref{simbol}) and by using (\ref{pa}), (\ref{bi}),
(\ref{ba}), (\ref{fin}) we have
\[
\int\nabla u_{0}\cdot\nabla v^{+}+\left[  \left\vert q\mathbf{A}_{0}%
-\ell\nabla\theta\right\vert ^{2}-\left(  q\phi_{0}-\omega\right)
^{2}\right]  \,u_{0}v^{+}+W^{\prime}\left(  u_{0}\right)  v^{+}=0\ .
\]

Taking $\varphi_{n}=v^{-}\chi_{n}$ and arguing in the same way as before, we
get
\[
\int\nabla u_{0}\cdot\nabla v^{-}+\left[  \left\vert q\mathbf{A}_{0}%
-\ell\nabla\theta\right\vert ^{2}-\left(  q\phi_{0}-\omega\right)
^{2}\right]  \,u_{0}v^{-}+W^{\prime}\left(  u_{0}\right)  v^{-}=0.
\]
Then
\[
\int\nabla u_{0}\cdot\nabla v+\left[  \left\vert q\mathbf{A}_{0}-\ell
\nabla\theta\right\vert ^{2}-\left(  q\phi_{0}-\omega\right)  ^{2}\right]
\,u_{0}v+W^{\prime}\left(  u_{0}\right)  v=0.
\]
Since $v\in\mathcal{D}$ is arbitrary, we get that equation (\ref{im}) is solved.
\end{proof}

The presence of the term $-\int\left\vert \nabla\phi\right\vert ^{2}$ gives to
the functional $J$ a strong\ indefiniteness, namely any critical point of $J$
has infinite Morse index: this fact is a great obstacle to a direct study of
the critical points. To avoid this difficulty we shall introduce a
\textit{reduced functional}

\subsection{The reduced functional}

Equation (\ref{z3}) can be written as follows
\begin{equation}
-\Delta\phi+q^{2}u^{2}\phi=q\omega u^{2}\label{a2bis}%
\end{equation}
and it can be easily verified (see \cite{bf}) that for any $u\in
H^{1}(\mathbb{R}^{3}),$ there exists a unique solution $\phi\in\mathcal{D}%
^{1,2}$ of (\ref{a2bis}).

Clearly, if $u\in\hat{H}_{r}^{1}(\mathbb{R}^{3}),$ the solution $\phi$
$=\phi_{u}$ of (\ref{a2bis}) belongs to $\mathcal{D}_{r}^{1,2}.$ Then we can
define the map
\begin{equation}
u\in\hat{H}_{r}^{1}(\mathbb{R}^{3})\rightarrow\ Z_{\omega}\left(  u\right)
=\phi_{u}\in\mathcal{D}_{r}^{1,2}\text{ solution of (\ref{a2bis}).}\label{map}%
\end{equation}
Standard arguments show that the map $Z_{\omega}$ is $C^{1}.$ Since $\phi_{u}$
solves (\ref{a2bis}), clearly we have
\begin{equation}
d_{\phi}J(u,Z_{\omega}\left(  u\right)  ,\mathbf{A)}=0\label{by}%
\end{equation}
where $J$ is defined in (\ref{functional}) and $d_{\phi}J$ denotes the partial
differential of $J$ with respect to $\phi.$ For $u\in H^{1}(\mathbb{R}^{3}),$
let $\Phi=\Phi_{u}$ be the solution of the equation (\ref{a2bis}) with
$\omega=1,$ then $\Phi_{u}$ solves the equation
\begin{equation}
-\Delta\Phi_{u}+q^{2}u^{2}\Phi_{u}=qu^{2}.\label{b2}%
\end{equation}
Clearly
\begin{equation}
\phi_{u}=\omega\Phi_{u}.\label{d}%
\end{equation}
Now let $q>0,$ then by maximum principle arguments it is easy to show
that\ for any $u\in H^{1}(\mathbb{R}^{3})$%
\begin{equation}
0\leq\Phi_{u}\leq\frac{1}{q}.\label{max}%
\end{equation}

Now, if $\left(  u,\mathbf{A}\right)  \in\hat{H}^{1}\times\left(
\mathcal{D}^{1,2}\right)  ^{3},$ we set
\[
\tilde{J}(u,\mathbf{A})=J(u,Z_{\omega}\left(  u\right)  ,\mathbf{A})
\]
where $J$ is defined in (\ref{functional}). By using the chain rule and
equation (\ref{by}), it is easy to verify (see the first part of the proof of
Theorem 16 in (\cite{befov07})) that
\begin{equation}
\left(  \left(  u,\mathbf{A}\right)  \text{ critical point of }\tilde
{J}\right)  \Longrightarrow\left(  (u,Z_{\omega}\left(  u\right)
,\mathbf{A})\text{ critical point of }J\right)  .\label{prem}%
\end{equation}

We will refer to $\tilde{J}(u,\mathbf{A})$ as the \textit{reduced action
functional}.

From (\ref{b2}) we have
\begin{equation}
\int qu^{2}\Phi_{u}dx=\int\left\vert \nabla\Phi_{u}\right\vert ^{2}%
dx+q^{2}\int u^{2}\Phi_{u}^{2}dx.\label{base}%
\end{equation}

Now, by (\ref{d}), (\ref{base}), we have:
\begin{align}
\tilde{J}(u,\mathbf{A})  & =J(u,Z_{\omega}\left(  u\right)  ,\mathbf{A}%
)=\frac{1}{2}\int\left\vert \nabla u\right\vert ^{2}-\left\vert \nabla\phi
_{u}\right\vert ^{2}+\left\vert \nabla\times\mathbf{A}\right\vert
^{2}\nonumber\\
& +\frac{1}{2}\int\left[  \left\vert \ell\nabla\theta-q\mathbf{A}\right\vert
^{2}-\left(  q\phi_{u}-\omega\right)  ^{2}\right]  \,u^{2}+\int W\left(
u\right) \nonumber\\
& =\frac{1}{2}\int\left(  \left\vert \nabla u\right\vert ^{2}+\left\vert
\nabla\times\mathbf{A}\right\vert ^{2}+\left\vert \ell\nabla\theta
-q\mathbf{A}\right\vert ^{2}u^{2}\right) \nonumber\\
& -\frac{1}{2}\omega^{2}\int\left(  \left\vert \nabla\Phi_{u}\right\vert
^{2}+q^{2}u^{2}\Phi_{u}^{2}+\,u^{2}-2qu^{2}\Phi_{u}\right)  +\int W\left(
u\right) \nonumber\\
& =\frac{1}{2}\int\left\vert \nabla u\right\vert ^{2}+\left\vert \nabla
\times\mathbf{A}\right\vert ^{2}+\left\vert \ell\nabla\theta-q\mathbf{A}%
\right\vert ^{2}u^{2}+\int W\left(  u\right) \nonumber\\
& -\frac{\omega^{2}}{2}\int(\left[  1-q{\Phi}_{u}\right]  )\,u^{2}%
.\label{prel}%
\end{align}
Then
\begin{equation}
\tilde{J}(u,\mathbf{A})=I(u,\mathbf{A})-\frac{\omega^{2}}{2}K_{q}%
(u)\label{prel1}%
\end{equation}

where
\[
I(u,\mathbf{A})=\frac{1}{2}\int\left\vert \nabla u\right\vert ^{2}+\left\vert
\nabla\times\mathbf{A}\right\vert ^{2}+\left\vert \ell\nabla\theta
-q\mathbf{A}\right\vert ^{2}u^{2}+\int W\left(  u\right)
\]
and
\begin{equation}
K_{q}(u)=\int(\left[  1-q{\Phi}_{u}\right]  )\,u^{2}.\label{chi}%
\end{equation}

Now, following the same lines as before, we can define the \textit{reduced
energy functional} as follows
\[
\mathcal{\tilde{E}}\left(  u,\mathbf{A}\right)  =\mathcal{E(}u,Z_{\omega
}\left(  u\right)  ,\mathbf{A)}%
\]
Where (see (\ref{enna}))
\begin{align}
\mathcal{E}  & =\frac{1}{2}\int\left(  \left\vert \nabla u\right\vert
^{2}+\left\vert \nabla\phi\right\vert ^{2}+\left\vert \nabla\times
\mathbf{A}\right\vert ^{2}+(\left\vert \ell\nabla\theta-q\mathbf{A}\right\vert
^{2}+\left(  \omega-q\phi\right)  ^{2})\,u^{2}\right) \nonumber\\
& +\int W(u).\label{enna1}%
\end{align}

It can be shown as for (\ref{prel1}) that
\begin{equation}
\mathcal{\tilde{E}}\left(  u,\mathbf{A}\right)  =I(u,\mathbf{A})+\frac
{\omega^{2}}{2}K_{q}(u).\label{prel2}%
\end{equation}

Observe that
\[
Q=q\sigma=q\omega K_{q}(u)
\]
represents the (electric) charge (see (\ref{ker}) and (\ref{car})), so that we
can write for $u\neq0$%
\[
\mathcal{\tilde{E}}\left(  u,\mathbf{A}\right)  =I(u,\mathbf{A})+\frac
{\omega^{2}}{2}K_{q}(u)=I(u,\mathbf{A})+\frac{\sigma^{2}}{2K_{q}(u)}.
\]
Then for any $\sigma$ $\neq0,$ the functional defined by
\begin{equation}
E_{\sigma,q}\left(  u,\mathbf{A}\right)  =I(u,\mathbf{A})+\frac{\sigma^{2}%
}{2K_{q}(u)},\text{ }\left(  u,\mathbf{A}\right)  \in\hat{H}^{1}\times\left(
\mathcal{D}^{1,2}\right)  ^{3},\text{ }u\neq0\label{esigma}%
\end{equation}

represents the energy on the configuration $\mathcal{(}u,\Phi_{u},\mathbf{A)}$
having charge $Q=q\sigma$ or equivalently frequency $\omega=\frac{\sigma
}{K_{q}(u)}.$

The following lemma holds

\begin{lemma}
\label{unno} Consider the functional
\[
\hat{H}^{1}\in u\rightarrow K(u)=\int\,u^{2}(1-q\Phi_{u})dx.
\]
Then for any $u\in\hat{H}^{1}$ we have%
\begin{equation}
K^{\prime}(u)=2u(1-q\Phi_{u})^{2}.\label{sec}%
\end{equation}
\bigskip
\end{lemma}

\begin{proof}
Set
\[
\mathcal{A(}u,\Phi)=\int\left\vert \nabla\Phi\right\vert ^{2}dx+\int
u^{2}(1-q\Phi)^{2}dx.
\]
By (\ref{base}) clearly we have
\[
\mathcal{A(}u,\Phi_{u})=K(u).
\]
Then
\begin{equation}
K^{\prime}(u)=\frac{\partial\mathcal{A}}{\partial u}\mathcal{(}u,\Phi
_{u})+\frac{\partial\mathcal{A}}{\partial\Phi}\mathcal{(}u,\Phi_{u}))\Phi
_{u}^{\prime}\label{give}%
\end{equation}
where $\frac{\partial\mathcal{A}}{\partial u},$ $\frac{\partial\mathcal{A}%
}{\partial\Phi}$ denote the partial derivatives of $\mathcal{A}$ with respect
to $u$ and $\Phi$ respectively. Since $\Phi_{u}$ solves (\ref{b2}), we have
\[
\frac{\partial\mathcal{A}}{\partial\Phi}\mathcal{(}u,\Phi_{u})=0.
\]
Then (\ref{give}) gives
\[
K^{\prime}(u)=\frac{\partial\mathcal{A}}{\partial u}\mathcal{(}u,\Phi
_{u})=2u(1-q\Phi_{u})^{2}.
\]

\end{proof}

The following proposition holds

\begin{proposition}
\label{fond}Let $\sigma\neq0$ and $\left(  u,\mathbf{A}\right)  \in\hat{H}%
^{1}\times\left(  \mathcal{D}^{1,2}\right)  ^{3},$ $u\neq0$ be a critical
point of $E_{\sigma,_{q}}$ (see (\ref{esigma})) . Then, if we set
$\omega=\frac{\sigma}{K_{q}(u)},$ $\mathcal{(}u,Z_{\omega}\left(  u\right)
,\mathbf{A)}$ is a critical point of $J$
\end{proposition}

\begin{proof}
Since $\left(  u,\mathbf{A}\right)  \in\hat{H}^{1}\times\left(  \mathcal{D}%
^{1,2}\right)  ^{3},$ $u\neq0$ is a critical point of $E_{\sigma,q},$ we have
\[
0=E_{\sigma,q}^{\prime}\left(  u,\mathbf{A}\right)  =I^{\prime}(u,\mathbf{A}%
)-\frac{\sigma^{2}K_{q}^{\prime}(u)}{2K_{q}(u)^{2}}=I^{\prime}(u,\mathbf{A}%
)-\frac{\omega^{2}K_{q}^{\prime}(u)}{2},\text{ }\omega=\frac{\sigma}{K_{q}%
(u)}.
\]
Hence $\left(  u,\mathbf{A}\right)  $ is a critical point of the functional
\[
\tilde{J}\left(  u,\mathbf{A}\right)  =I(u,\mathbf{A})-\frac{\omega^{2}%
K_{q}(u)}{2}.
\]
So by (\ref{prem}) $\left(  u,Z_{\omega}\left(  u\right)  ,\mathbf{A}\right)
$ is a critical point of $J.$
\end{proof}

By Proposition \ref{fond} and Theorem \ref{finale} we are reduced to study the
critical points of $E_{\sigma,q}$ which is a functional bounded from below.

However $E_{\sigma,q}$ contains the term $\int\left\vert \nabla\times
\mathbf{A}\right\vert ^{2}$ which is not a Sobolev norm.

In order to avoid this difficulty we introduce a suitable manifold
$V\subset\hat{H}^{1}\times\left(  \mathcal{D}^{1,2}\right)  ^{3}$ such that:

\begin{itemize}
\item the critical points of $J$ restricted to $V$ satisfy Eq. (\ref{z1}),
(\ref{z3}), \ref{z4}); namely $V$ is a \textquotedblright natural
constraint\textquotedblright\ for $J$.

\item The components $\mathbf{A}$ of the elements in $V$ are divergence free,
then the term $\int\left\vert \nabla\times\mathbf{A}\right\vert ^{2} $ can be
replaced by $\left\Vert \mathbf{A}\right\Vert _{\left(  \mathcal{D}%
^{1,2}\right)  ^{3}}^{2}=\int\left\vert \nabla\mathbf{A}\right\vert ^{2}$.
\end{itemize}

We set
\begin{equation}
\mathcal{A}_{0}:=\left\{  \mathbf{X}\in\mathcal{C}_{0}^{\infty}(\mathbb{R}%
^{3}\backslash\Sigma,\mathbb{R}^{3}):\mathbf{X}=b\left(  r,x_{3}\right)
\nabla\theta;\ b\in C_{0}^{\infty}\left(  \mathbb{R}^{3}\backslash
\Sigma,\mathbb{R}\right)  \right\}  .\label{ac}%
\end{equation}
Let $\mathcal{A}$ denote the closure of $\mathcal{A}_{0}$ with respect to the
norm of $\left(  \mathcal{D}^{1,2}\right)  ^{3}.$ \ We shall consider the
following space
\begin{equation}
V:=\hat{H}_{r}^{1}\times\mathcal{A}\label{set}%
\end{equation}
where $\hat{H}_{r}^{1},$ is the closure of $\mathcal{D}_{r}$ with respect to
the $\hat{H}^{1}$ norm. We shall set $U=\left(  u,\mathbf{A}\right)  $ and
\[
\left\Vert U\right\Vert _{V}=\left\Vert \left(  u,\mathbf{A}\right)
\right\Vert _{V}=\left\Vert u\right\Vert _{\hat{H}_{r}^{1}}+\left\Vert
\mathbf{A}\right\Vert _{\left(  \mathcal{D}^{1,2}\right)  ^{3}}.
\]

\begin{lemma}
\label{div}If $\mathbf{A}\in\mathcal{A}$, then
\[
\int\left\vert \nabla\times\mathbf{A}\right\vert ^{2}=\int\left\vert
\nabla\mathbf{A}\right\vert ^{2}.
\]

\end{lemma}

\begin{proof}
Let $\mathbf{A=}b\nabla\theta\in\mathcal{A}_{0}.$ Since $b\ $depends only on
$r\ $and $x_{3},$\ it\ is easy to check that
\[
\nabla b\cdot\nabla\theta=0.
\]
Since $\theta$ is harmonic in $\mathbb{R}^{3}\backslash\Sigma$ and $b$ has
support in $\mathbb{R}^{3}\backslash\Sigma$%
\[
b\Delta\theta=0.
\]
Then
\[
\nabla\cdot\mathbf{A=}\nabla\cdot\left(  b\nabla\theta\right)  =\nabla
b\cdot\nabla\theta+b\Delta\theta=0.
\]
Thus, by continuity, we get
\[
\int\left(  \nabla\cdot\mathbf{A}\right)  ^{2}=0\text{ for any }\mathbf{A}%
\in\mathcal{A}\text{ .}%
\]
Then
\[
\int\left\vert \nabla\times\mathbf{A}\right\vert ^{2}=\int\left(  \nabla
\cdot\mathbf{A}\right)  ^{2}+\int\left\vert \nabla\times\mathbf{A}\right\vert
^{2}=\int\left\vert \nabla\mathbf{A}\right\vert ^{2}.
\]

\end{proof}

\subsection{Analysis of the minimizing sequences}

The ratio energy/charge is a crucial quantity for the following lemmas. For a
charge $\sigma>0$ this ratio is defined as function of $u$ and $\mathbf{A} $
in the following way
\[
\Lambda_{\sigma,q}\left(  u,\mathbf{A}\right)  =\frac{E_{\sigma,q}%
(u,\mathbf{A})}{\sigma}=\frac{I(u,\mathbf{A})}{\sigma}+\frac{\sigma}%
{2K_{q}(u)},\text{ }\left(  u,\mathbf{A}\right)  \in\hat{H}^{1}\times\left(
\mathcal{D}^{1,2}\right)  ^{3},\text{ }u\neq0
\]
where
\begin{equation}
K_{q}(u)=\int(\left[  1-q{\Phi}_{u}\right]  )\,u^{2}.\label{qu}%
\end{equation}
In the following we shall always assume that the $W$ satisfies W1),W2),W3).
Firtst we state the following continuity lemma:

\begin{lemma}
\label{continuity} Let $u\in H^{1},$ then%
\[%
{\displaystyle\int}
(1-q\Phi_{u})u^{2}\rightarrow%
{\displaystyle\int}
u^{2}\text{ as }q\rightarrow0
\]

\end{lemma}

\begin{proof}
Clearly it is enough to show that
\begin{equation}
q\int\Phi_{u}u^{2}\rightarrow0\text{ as }q\rightarrow0\label{fina}%
\end{equation}
Since $\Phi_{u}$ depends on $q$ a little work is needed to prove (\ref{fina}).
Since $\Phi_{u}$ solves (\ref{b2}), we have%
\[
\left\Vert \Phi_{u}\right\Vert _{\mathcal{D}^{1,2}}^{2}+q^{2}\int u^{2}%
\Phi_{u}^{2}=q\int u^{2}\Phi_{u}\leq
\]%
\begin{equation}
\leq q\left\Vert u\right\Vert _{L^{\frac{12}{5}}}^{2}\left\Vert \Phi
_{u}\right\Vert _{L^{6}}%
\end{equation}
and then, if $u\neq0,$ we have
\[
\frac{\left\Vert \Phi_{u}\right\Vert _{\mathcal{D}^{1,2}}^{2}}{\left\Vert
\Phi_{u}\right\Vert _{L^{6}}}\leq q\left\Vert u\right\Vert _{L^{\frac{12}{5}}%
}^{2}.
\]
So, since $\mathcal{D}^{1,2}$ is continuously embedded into $L^{6}$, we easily
get
\begin{equation}
\left\Vert \Phi_{u}\right\Vert _{\mathcal{D}^{1,2}}\leq c_{1}q\left\Vert
u\right\Vert _{L^{\frac{12}{5}}}^{2},
\end{equation}
where $c_{1}$ is a positive constant. Then we get
\[
q\int u^{2}\Phi_{u}\leq q\left\Vert u\right\Vert _{L^{\frac{12}{5}}}%
^{2}\left\Vert \Phi_{u}\right\Vert _{L^{6}}\leq c_{1}q^{2}\left\Vert
u\right\Vert _{L^{\frac{12}{5}}}^{4}.
\]
From which we deduce (\ref{fina}).
\end{proof}

\begin{lemma}
\label{au} There exist $\sigma,$ $\bar{q}>0,$ such that for all $0\leq
q<\bar{q}$ there exists $u\in\hat{H}_{r}^{1}$ such that
\[
\Lambda_{\sigma,q}(u,0)<1.
\]

\end{lemma}

\begin{proof}
For $0<\mu<\lambda$ we set:
\[
T_{\lambda,\mu}=\left\{  \left(  r,x_{3}\right)  :(r-\lambda)^{2}+x_{3}{}%
^{2}\leq\mu\right\}
\]
and, for $\lambda>2,$ we consider a smooth function $u_{\lambda}$ with
cylindrical symmetry such that
\[
u_{\lambda}(r,x_{3})=\left\{
\begin{array}
[c]{cc}%
s_{0} & if\;\;\left(  r,x_{3}\right)  \in T_{\lambda,\lambda/2}\\
& \\
0 & if\;\;\left(  r,x_{3}\right)  \notin T_{\lambda,\lambda/2+1}%
\end{array}
\right.
\]
where $s_{0}$ is such that $N(s_{0})<0$ (see (\ref{pec}))$.$ Moreover we may
assume that
\[
\left\vert \nabla u_{\lambda}\left(  r,x_{3}\right)  \right\vert
\leq2\;\text{for }\left(  r,x_{3}\right)  \in T_{\lambda,\lambda
/2+1}\backslash T_{\lambda,\lambda/2}.
\]

We have that for all $\sigma\neq0$
\begin{align*}
\Lambda_{\sigma,q}(u_{\lambda},0)  & =\frac{1}{\sigma}\int\left[  \frac{1}%
{2}\left\vert \nabla u_{\lambda}\right\vert ^{2}+\frac{\ell^{2}}{2}%
\frac{u_{\lambda}^{2}}{r^{2}}+W(u_{\lambda})\right]  dx+\frac{\sigma}%
{2K_{q}(u_{\lambda})}\\
& =\frac{\int\left[  \left\vert \nabla u_{\lambda}\right\vert ^{2}+\frac
{\ell^{2}u_{\lambda}^{2}}{r^{2}}\right]  dx}{2\sigma}+\frac{\int u_{\lambda
}^{2}}{2\sigma}+\frac{\int N(u_{\lambda})dx}{\sigma}+\frac{\sigma}%
{2K_{q}(u_{\lambda})}.
\end{align*}
Now take
\[
\sigma=\sigma_{\lambda}=\int u_{\lambda}^{2}%
\]
in this case we get
\begin{equation}
\Lambda_{\sigma_{\lambda},q}(u_{\lambda},0)=\frac{1}{2}+\frac{\sigma_{\lambda
}}{2K_{q}(u_{\lambda})}+\frac{\int\left[  \left\vert \nabla u_{\lambda
}\right\vert ^{2}+\frac{\ell^{2}u_{\lambda}^{2}}{r^{2}}\right]  dx}{2\int
u_{\lambda}^{2}}+\frac{\int N(u_{\lambda})dx}{\int u_{\lambda}^{2}%
}.\label{lan}%
\end{equation}

By a direct computation we have that
\begin{equation}
\int\left\vert \nabla u_{\lambda}\right\vert ^{2}\leq c_{1}meas(T_{\lambda
,\lambda/2+1}\backslash T_{\lambda,\lambda/2})=c_{2}\lambda^{2}\label{al}%
\end{equation}%
\begin{equation}
\int\frac{u_{\lambda}^{2}}{r^{2}}\leq\frac{c_{3}}{\lambda^{2}}meas(T_{\lambda
,\lambda/2+1})=c_{4}\lambda\label{all}%
\end{equation}%
\begin{equation}
\int u_{\lambda}^{2}\geq c_{5}meas(T_{\lambda,\lambda/2+1})=c_{6}\lambda
^{3}.\label{alll}%
\end{equation}

So that
\begin{equation}
\frac{\int\left[  \left\vert \nabla u_{\lambda}\right\vert ^{2}+\frac{\ell
^{2}u_{\lambda}^{2}}{r^{2}}\right]  dx}{2\int u_{\lambda}^{2}}=O\left(
\frac{1}{\lambda}\right)  .\label{pin}%
\end{equation}

Moreover
\[
\int N(u_{\lambda})dx\leq N(s_{0})meas(T_{\lambda,\lambda/2})+c_{7}%
meas(T_{\lambda,\lambda/2+1}\backslash T_{\lambda,\lambda/2})=
\]%
\begin{equation}
\leq c_{8}N(s_{0})\lambda^{3}+c_{9}\lambda^{2}.\label{lumi}%
\end{equation}
From (\ref{lumi}) and (\ref{alll}) we get
\begin{equation}
\frac{\int N(u_{\lambda})dx}{\int u_{\lambda}^{2}}\leq c_{10}\frac{N(s_{0}%
)}{s_{0}^{2}}+O\left(  \frac{1}{\lambda}\right)  =g(s_{0},\lambda
).\label{cont}%
\end{equation}
From (\ref{lan}), (\ref{pin}) and (\ref{cont}) we get
\begin{equation}
\Lambda_{\sigma_{\lambda},q}(u_{\lambda},0)\leq\frac{1}{2}+\frac
{\sigma_{\lambda}}{2K_{q}(u_{\lambda})}+g(s_{0},\lambda).\label{anc}%
\end{equation}

Since $N(s_{0})<0,$ we can take $\lambda_{0}$ so large that
\begin{equation}
g(s_{0},\lambda_{0})<0.\label{meno}%
\end{equation}

Now we take
\[
\sigma=\sigma_{\lambda_{0}}=\int u_{\lambda_{0}}^{2}\text{, and }%
u=u_{\lambda_{0}}.
\]
Now, by Lemma \ref{continuity}, we have
\[
K_{q}(u)\rightarrow K_{0}(u)=\sigma\text{ for }q\rightarrow0.
\]
So%
\begin{equation}
\frac{\sigma}{2K_{q}(u)}\rightarrow\frac{1}{2}\text{ for }q\rightarrow
0.\label{lam}%
\end{equation}
Then, by (\ref{anc}), (\ref{meno}) and (\ref{lam}), there is $\bar{q}>0$ so
small that, for all $\ 0\leq q<\bar{q},$ we have%
\[
\Lambda_{\sigma,q}(u,0)\leq\frac{1}{2}+\frac{\sigma}{2K_{q}(u)}+g(s_{0}%
,\lambda_{0})<1.
\]

\end{proof}

Now the following a priori estimate on the minimizing sequences can be obtained

\begin{lemma}
\label{tec}Any minimizing sequence $\left(  u_{n},\mathbf{A}_{n}\right)
\subset V$ for $E_{\sigma,q}\mid_{V}$ is bounded in $\hat{H}^{1}\times\left(
\mathcal{D}^{1,2}\right)  ^{3}.$
\end{lemma}

\begin{proof}
Let $\left(  u_{n},\mathbf{A}_{n}\right)  \subset V$ be a minimizing sequence
for $E_{\sigma,q}\mid_{V}.$ Clearly
\[
\left\Vert \mathbf{A}_{n}\right\Vert _{\left(  \mathcal{D}^{1,2}\right)  ^{3}%
}\text{ is bounded.}%
\]
So it remain to prove that
\begin{equation}
\left\Vert u_{n}\right\Vert _{\hat{H}_{r}^{1}}\text{ is bounded.}\label{bou}%
\end{equation}

To this end we shall first show that
\begin{equation}
\left\Vert u_{n}\right\Vert _{L^{2}}\text{ is bounded.}\label{bouu}%
\end{equation}

Since $\left(  u_{n},\mathbf{A}_{n}\right)  $ is a minimizing sequence for
$E_{\sigma,q}\mid_{V}$ we get
\begin{equation}
\int W(u_{n})\text{ and }\int\left\vert \nabla u_{n}\right\vert ^{2}\text{ are
bounded.}\label{lim}%
\end{equation}
Then we have also that
\begin{equation}
\int u_{n}^{6}\text{ is bounded.}\label{limm}%
\end{equation}

Let $\varepsilon>0$ and set
\[
\Omega_{n}=\left\{  x\in\mathbb{R}^{3}:\left\vert u_{n}(x)\right\vert
>\varepsilon\right\}  \text{ and }\Omega_{n}^{c}=\mathbb{R}^{3}\backslash
\Omega_{n}.
\]
By (\ref{lim}) and since $W\geq0$ we have
\begin{equation}
\int_{\text{ }\Omega_{n}^{c}}W(u_{n})\text{ is bounded .}\label{llim}%
\end{equation}

By $W_{2})$ we can write
\[
W(s)=\frac{1}{2}s^{2}+\circ(s^{2})\text{.}%
\]
Then, if $\varepsilon$ is small enough, there is a constant $c>0$ such that
\begin{equation}
\int_{\text{ }\Omega_{n}^{c}}W(u_{n})\geq c\int_{\Omega_{n}^{c}}u_{n}%
^{2}.\label{ma}%
\end{equation}
By (\ref{llim}) and (\ref{ma}) we get that
\begin{equation}
\int_{\Omega_{n}^{c}}u_{n}^{2}\text{ is bounded.}\label{pi}%
\end{equation}
On the other hand
\begin{equation}
\int_{\Omega_{n}}u_{n}^{2}\leq\left(  \int_{\Omega_{n}}u_{n}^{6}\right)
^{\frac{1}{3}}meas(\Omega_{n})^{\frac{2}{3}}.\label{u}%
\end{equation}

By (\ref{limm}) we have that
\begin{equation}
meas(\Omega_{n})\text{ is bounded.}\label{uu}%
\end{equation}

By (\ref{u}), (\ref{uu}) and again by (\ref{limm}) we get
\begin{equation}
\int_{\Omega_{n}}u_{n}^{2}\text{ is bounded.}\label{ff}%
\end{equation}
So (\ref{bouu}) follows from (\ref{pi}) and (\ref{ff}).

Let us finally prove (\ref{bou}).

Clearly
\[
E_{\sigma,q}\left(  u_{n},\mathbf{A}_{n}\right)  \geq I\left(  u_{n}%
,\mathbf{A}_{n}\right)  \geq
\]%
\[
\frac{1}{2}\int\left(  \left\vert \nabla u_{n}\right\vert ^{2}+\left\vert
\nabla\mathbf{A}_{n}\right\vert ^{2}+q^{2}\left\vert \mathbf{A}_{n}\right\vert
^{2}u_{n}^{2}+\ell^{2}\frac{u_{n}^{2}}{r^{2}}-2q\frac{\ell}{r}\left\vert
\mathbf{A}_{n}\right\vert \left\vert u_{n}\right\vert ^{2}\right)  dx\geq
\]%
\begin{equation}
\frac{1}{2}\left\Vert u_{n}\right\Vert _{\hat{H}_{r}^{1}}^{2}-q%
{\displaystyle\int}
\frac{\ell}{r}\left\vert \mathbf{A}_{n}\right\vert \left\vert u_{n}\right\vert
^{2}-\sup\left\Vert u_{n}\right\Vert _{L^{2}}\label{co}%
\end{equation}

Also we have%
\[%
{\displaystyle\int}
\frac{q\ell}{r}\left\vert \mathbf{A}_{n}\right\vert \left\vert u_{n}%
\right\vert ^{2}\leq\frac{1}{2}%
{\displaystyle\int}
\left(  4q^{2}\ell^{2}\left\vert \mathbf{A}_{n}\right\vert ^{2}+\frac
{1}{4r^{2}}\right)  \left\vert u_{n}\right\vert ^{2}\leq
\]%
\begin{equation}
\frac{1}{8}\left\Vert u_{n}\right\Vert _{\hat{H}_{r}^{1}}^{2}+2q^{2}\ell^{2}%
{\displaystyle\int}
\left\vert \mathbf{A}_{n}\right\vert ^{2}\left\vert u_{n}\right\vert
^{2}.\label{coco}%
\end{equation}

Since $E_{\sigma,q}\left(  u_{n},\mathbf{A}_{n}\right)  $ is bounded, by
(\ref{co}) and (\ref{coco}) we deduce that%
\begin{equation}
c_{1}\geq\left(  \frac{1}{2}-\frac{1}{8}\right)  \left\Vert u_{n}\right\Vert
_{\hat{H}_{r}^{1}}^{2}-2q^{2}\ell^{2}%
{\displaystyle\int}
\left\vert \mathbf{A}_{n}\right\vert ^{2}\left\vert u_{n}\right\vert
^{2}.\label{apo}%
\end{equation}
Here $c_{1},c_{2}$ will denote suitable constants.

Now, since $\left\Vert u_{n}\right\Vert _{L^{2}}$ and $\left\Vert
u_{n}\right\Vert _{L^{6}\text{ }}$are bounded, also $\left\Vert u_{n}%
\right\Vert _{L^{3}\text{ }}$is bounded.

Then, by using also the boundeness of $\left\Vert \mathbf{A}_{n}\right\Vert
_{L^{6}},$ we get
\begin{equation}%
{\displaystyle\int}
\left\vert \mathbf{A}_{n}\right\vert ^{2}\left\vert u_{n}\right\vert ^{2}%
\leq\left(  \left\Vert \mathbf{A}_{n}\right\Vert _{L^{6}}\right)  ^{\frac
{1}{3}}\left(  \left\Vert u_{n}\right\Vert _{L^{3}\text{ }}\right)  ^{\frac
{2}{3}}\leq c_{2}.\label{quu}%
\end{equation}
From (\ref{apo}) and (\ref{quu}) we deduce the boundeness of $\left\Vert
u_{n}\right\Vert _{\hat{H}_{r}^{1}}^{2}$.
\end{proof}

By Lemma \ref{tec} any minimizing sequence $U_{n}:=\left(  u_{n}%
,\mathbf{A}_{n}\right)  \subset V$ of $E_{\sigma,q}\mid_{V}$ weakly converges
(up to a subsequence). Observe that $E_{\sigma,q}$ is invariant for
translations along the $x_{3}$-axis, namely for $U\in V$ and $L\in\mathbb{R}$
we have
\[
E_{\sigma,q}(T_{L}U)=E_{\sigma,q}(U)
\]
where
\begin{equation}
T_{L}\left(  U\right)  \left(  x_{1},x_{2},x_{3}\right)  =U\left(  x_{1}%
,x_{2},x_{3}+L\right)  .\label{trasla}%
\end{equation}

As consequence of this invariance we have that $\left(  u_{n},\mathbf{A}%
_{n}\right)  $ does not contain in general a (strongly) convergent
subsequence. So we argue as follows: we prove that for suitable $\sigma,$ $q $
there exists a minimizing sequence $\left(  u_{n},\mathbf{A}_{n}\right)  $ of
$E_{\sigma,q}\mid_{V}$ which, up to translations along the $x_{3}$-direction,
weakly converges to a non trivial limit $\left(  u_{0},\mathbf{A}_{0}\right)
.$ This limit will be actually a critical point of $E_{\sigma_{0},q}$ for some
charge $\sigma_{0}.$

To follow the above program we first prove the following Lemma

\begin{lemma}
\label{mi}Let $U_{n}=\left(  u_{n},\mathbf{A}_{n}\right)  \subset V$ be a
minimizing sequence of $E_{\sigma,q}\mid_{V},$ $\sigma>0.$ Then there exist
$\delta,M>0$ such that%
\[
\delta\leq\omega_{n}\leq M
\]
where
\[
\omega_{n}=\frac{\sigma}{K_{q}(u_{n})}.
\]

\end{lemma}

\begin{proof}
\bigskip\ Since $\left(  u_{n},\mathbf{A}_{n}\right)  \subset V$ is a
minimizing sequence of the functional $E_{\sigma,q}\mid_{V}$ defined by
\[
E_{\sigma,q}\left(  u,\mathbf{A}\right)  =I(u,\mathbf{A})+\frac{\sigma^{2}%
}{2K_{q}(u)},
\]
we have that for some constant $c_{1}>0$
\begin{equation}
c_{1}\leq K_{q}(u_{n}).\label{s}%
\end{equation}

Also for some constant $c_{2}>0$ we have
\begin{equation}
K_{q}(u_{n})\text{ }\leq c_{2}.\label{ss}%
\end{equation}

In fact, arguing by contradiction, we assume that, up to a subsequence%
\[
K_{q}(u_{n})=\int(\left[  1-q{\Phi}_{u_{n}}\right]  )\,u_{n}^{2}%
\rightarrow\infty,
\]
then by (\ref{max}) also we get%
\[
\int\,u_{n}^{2}\rightarrow\infty
\]
contradicting (\ref{bouu}).

Finally the conclusion immediately follows from (\ref{s}) and (\ref{ss}).
\end{proof}

Now we shall prove the following proposition

\begin{proposition}
\label{on} There exist $\sigma,$ $\bar{q}>0$ such that for all $0\leq
q<\bar{q},$ for any minimizing sequence $\left(  u_{n},\mathbf{A}_{n}\right)
\subset V$ of $E_{\sigma,q}\mid_{V}$ we have
\[
\left\Vert u_{n}\right\Vert _{L^{3}}\geq c>0\text{ for }n\text{ large.}%
\]

\end{proposition}

\begin{proof}
Let $\sigma$ and $q$ be chosen as required in Lemma \ref{au}. Now let $\left(
u_{n},\mathbf{A}_{n}\right)  \subset V$ be a minimizing sequence of
$E_{\sigma,q}$ and hence of $\Lambda_{\sigma,q}$. Then by Lemma\ref{au} we
get
\begin{equation}
\Lambda_{\sigma,q}(u_{n},\mathbf{A}_{n})\leq1-\delta,\;\delta>0\label{mil}%
\end{equation}
Then we have also%

\[
\frac{\int\left[  \left\vert \nabla u_{n}\right\vert ^{2}+\frac{\ell^{2}%
u_{n}^{2}}{r^{2}}\right]  dx}{2\sigma}+\frac{\int u_{n}^{2}}{2\sigma}%
+\frac{\int N(u_{n})dx}{\sigma}+\frac{\sigma}{2\int u_{n}^{2}}\leq1-\delta.
\]
Thus
\[
\frac{\int N(u_{n})dx}{\sigma}\leq1-\delta-\left(  \frac{\int u_{n}^{2}%
}{2\sigma}+\frac{\sigma}{2\int u_{n}^{2}}\right)  \leq-\delta.
\]
This implies that
\[
\int N(u_{n})dx\leq-\delta\sigma.
\]
On the other hand, by the assumptions on $W,$ we have that
\[
N(s)\geq-bs^{3},\text{ }b>0
\]
Then
\[
b\int\left\vert u_{n}\right\vert ^{3}\geq-\int N(u_{n})dx\geq\delta\sigma.
\]
\bigskip
\end{proof}

\begin{proposition}
\label{quasi}For any $\sigma,$ $q\geq0$ there exists a minimizing sequence
$\left(  u_{n},\mathbf{A}_{n}\right)  $ of $E_{\sigma,q}\mid_{V},$ with
$u_{n}\geq0$ and which is also a P.S. sequence for $E_{\sigma,q},$ i.e.
\[
E_{\sigma,q}^{\prime}\left(  u_{n},\mathbf{A}_{n}\right)  \rightarrow0.
\]

\end{proposition}

\begin{proof}
\textbf{\ }Let $\left(  u_{n},\mathbf{A}_{n}\right)  \subset V$ be a
minimizing sequence for $E_{\sigma}\mid_{V}.$ It is not restrictive to assume
that $u_{n}\geq0,$ in fact, if not, we can replace $u_{n}$ with $\left\vert
u_{n}\right\vert $ (see (\ref{enna1})). By standard variational arguments we
can also assume that $\left(  u_{n},\mathbf{A}_{n}\right)  $ is a P.S.
sequence for $E_{\sigma}\mid_{V},$ namely we can assume that
\[
E_{\sigma,q}^{\prime}\mid_{V}\left(  u_{n},\mathbf{A}_{n}\right)
\rightarrow0.
\]
By using the same arguments in proving Theorem 16 in \cite{befov07}, it can be
shown that $\left(  u_{n},\mathbf{A}_{n}\right)  $ is a P.S. sequence also for
$E_{\sigma,q},$ i.e.
\begin{equation}
E_{\sigma,q}^{\prime}\left(  u_{n},\mathbf{A}_{n}\right)  \rightarrow
0.\label{pal}%
\end{equation}

\end{proof}

\begin{proposition}
\label{la}There exist $\sigma$, $\bar{q}>0$ such that for all $0\leq q<\bar
{q}$ there exists a P.S. sequence $U_{n}=\left(  u_{n},\mathbf{A}_{n}\right)
$ for $E_{\sigma,q}$ which weakly converges to $\left(  u_{0},\mathbf{A}%
_{0}\right)  ,$ $u_{0}\geq0$ and $u_{0}\neq0.$.
\end{proposition}

\begin{proof}
Take $\sigma$, $q$ as in Proposition \ref{on}. By Proposition \ref{quasi}
there exists a minimizing sequence $U_{n}=\left(  u_{n},\mathbf{A}_{n}\right)
$ of $E_{\sigma,q}\mid_{V}$ with $u_{n}\geq0$ and which is also a P.S.
sequence for $E_{\sigma,q},$ i.e.
\[
E_{\sigma,q}^{\prime}\left(  U_{n}\right)  \rightarrow0.
\]
By Proposition \ref{on} we can assume that%
\[
\left\Vert u_{n}\right\Vert _{L^{3}}\geq c>0\text{ for }n\text{ large.}%
\]
By Lemma \ref{tec} the sequence $\left\{  U_{n}\right\}  $ is bounded in
$\hat{H}^{1}\times\left(  \mathcal{D}^{1,2}\right)  ^{3}$ so we can assume
that it weakly converges. However the weak limit could be trivial. We will
show that there is a sequence of integers $j_{n}$ such that (up to a
subsequence) $V_{n}:=T_{j_{n}}U_{n}\rightharpoonup U_{0}=\left(
u_{0},\mathbf{A}_{0}\right)  ,$ $u_{0}\neq0,$ weakly in $H^{1}\times\left(
\mathcal{D}^{1,2}\right)  ^{3}$.

We set
\[
\Omega_{j}=\left\{  \left(  x_{1},x_{2},x_{3}\right)  :j\leq x_{3}%
<j+1\right\}  \text{, }j\text{ integer}%
\]
In the following $c_{1},...,c_{4}$ denote positive constants.

We have for $n$ large
\begin{align*}
0  & <c_{1}\leq\left\Vert u_{n}\right\Vert _{L^{3}}^{3}=\sum_{j}\int
_{\Omega_{j}}\left\vert u_{n}\right\vert ^{3}=\sum_{j}\left(  \int_{\Omega
_{j}}\left\vert u_{n}\right\vert ^{3}\right)  ^{1/3}\cdot\left(  \int
_{\Omega_{j}}\left\vert u_{n}\right\vert ^{3}\right)  ^{2/3}\\
& \leq\sup_{j}\left\Vert u_{n}\right\Vert _{L^{3}\left(  \Omega_{j}\right)
}\sum_{j}\left(  \int_{\Omega_{j}}\left\vert u_{n}\right\vert ^{3}\right)
^{2/3}\leq c_{2}\cdot\sup_{j}\left\Vert u_{n}\right\Vert _{L^{3}\left(
\Omega_{j}\right)  }\cdot\sum_{j}\left\Vert u_{n}\right\Vert _{H^{1}\left(
\Omega_{j}\right)  }^{2}\\
& \leq c_{2}\cdot\sup_{j}\left\Vert u_{n}\right\Vert _{L^{3}\left(  \Omega
_{j}\right)  }\cdot\left\Vert u_{n}\right\Vert _{H^{1}\left(  \mathbb{R}%
^{3}\right)  }^{2}\leq(\text{since }\left\Vert u_{n}\right\Vert _{H^{1}\left(
\mathbb{R}^{3}\right)  }^{2}\leq c_{3})\text{ }\\
& \leq c_{2}c_{3}\sup_{j}\left\Vert u_{n}\right\Vert _{L^{3}\left(  \Omega
_{j}\right)  }.
\end{align*}
Then, for $n$ large$,$ there exists an integer $j_{n}$ such that
\begin{equation}
\left\Vert u_{n}\right\Vert _{L^{3}\left(  \Omega_{j_{n}}\right)  }\geq
\frac{c_{1}}{2c_{2}c_{3}}:=c_{4}>0.\label{paracula}%
\end{equation}
Now set
\[
\left(  u_{n}^{\prime},\mathbf{A}_{n}^{^{\prime}}\right)  =\ U_{n}^{\prime
}(x_{1},x_{2},x_{3})=U_{n}(x_{1},x_{2},x_{3}+j_{n})=T_{j_{n}}\left(
U_{n}\right)
\]
By Lemma \ref{tec} the sequence $u_{n}^{\prime}$ is bounded $\hat{H}%
^{1}\left(  \mathbb{R}^{3}\right)  ,$ then (up to a subsequence) it converges
weakly to $u_{0}\in\hat{H}^{1}\left(  \mathbb{R}^{3}\right)  .$ Clearly
$u_{0}\geq0$, since $u_{n}^{\prime}\geq0.$ We want to show that $u_{0}\neq0.$
Now, let $\varphi=\varphi\left(  x_{3}\right)  $ be a nonnegative, $C^{\infty
}$-function whose value is $1$ for $0<x_{3}<1$ and $0 $ for $\left\vert
x_{3}\right\vert >2.$ Then, the sequence $\varphi u_{n}^{\prime}$ is bounded
in $H_{0}^{1}(\mathbb{R}^{2}\times(-2,2)),$ moreover $\varphi u_{n}^{\prime}$
has cylindrical symmetry. Then, using the compactness result proved in
\cite{el}, we have
\[
\varphi u_{n}^{\prime}\rightarrow\chi\text{ strongly in }L^{3}(\mathbb{R}%
^{2}\times(-2,2)).
\]
On the other hand
\begin{equation}
\varphi u_{n}^{\prime}\rightarrow\varphi u_{0}\text{ }a.e\text{.}\label{punto}%
\end{equation}

Then
\begin{equation}
\varphi u_{n}^{\prime}\rightarrow\varphi u_{0}\text{ strongly in }%
L^{3}(\mathbb{R}^{2}\times(-2,2)).\label{virgola}%
\end{equation}

Moreover by (\ref{paracula})
\begin{equation}
\left\Vert \varphi u_{n}^{\prime}\right\Vert _{L^{3}\left(  \mathbb{R}%
^{2}\times(-2,2)\right)  }\geq\left\Vert u_{n}^{\prime}\right\Vert
_{L^{3}\left(  \Omega_{0}\right)  }=\left\Vert u_{n}\right\Vert _{L^{3}\left(
\Omega_{j_{n}}\right)  }\geq c_{4}.\label{pp}%
\end{equation}

Then$\ $by (\ref{virgola}) and (\ref{pp})
\[
\left\Vert \varphi u_{0}\right\Vert _{L^{3}\left(  \mathbb{R}^{2}%
\times(-2,2)\right)  }\geq c_{4}>0.
\]
Thus we have that $u_{0}\neq0.$
\end{proof}

\begin{proposition}
\label{sol} There exists $\bar{q}>0$ such that, for all $0\leq q<\bar{q},$ for
some charge $\sigma_{0}>0,$ $E_{\sigma_{0},q}$ has a critical point $\left(
u_{0},\mathbf{A}_{0}\right)  $ $u_{0}\neq0$, $u_{0}\geq0$.
\end{proposition}

\begin{proof}
Let $\sigma$, $q>0$ be as in Proposition \ref{la}, then there exists a
sequence $U_{n}=\left(  u_{n},\mathbf{A}_{n}\right)  $ in $V$, with $u_{n}%
\geq0$ and such that%
\begin{equation}
E_{\sigma,q}^{\prime}\left(  u_{n},\mathbf{A}_{n}\right)  \rightarrow
0\label{pus}%
\end{equation}

and
\[
\left(  u_{n},\mathbf{A}_{n}\right)  \rightarrow\left(  u_{0},\mathbf{A}%
_{0}\right)  \text{ weakly, }u_{0}\neq0
\]
Since $u_{n}\geq0$ we have $u_{0}\geq0$ .

Let us show that $U_{0}=\left(  u_{0},\mathbf{A}_{0}\right)  $ is a critical
point of $E_{\sigma_{0},q}$ for some charge $\sigma_{0}>0.$

By (\ref{pus}) we get that%
\[
dE_{\sigma,q}\left(  U_{n}\right)  \left[  w,0\right]  \rightarrow0\text{ and
}dE_{\sigma,q}\left(  U_{n}\right)  \left[  0,\mathbf{w}\right]
\rightarrow0\text{ for any }\left(  w,\mathbf{w}\right)  \in\hat{H}^{1}%
\times\left(  \mathcal{D}^{1,2}\right)  ^{3}.
\]

Then for any $w\in C_{0}^{\infty}(\mathbb{R}^{3}\backslash\Sigma)$ and
$\mathbf{w}\in\left(  C_{0}^{\infty}(\mathbb{R}^{3})\right)  ^{3}$ we have
\begin{equation}
d_{u}I(U_{n})\left[  w\right]  +d_{u}\left(  \frac{\sigma^{2}}{2K_{q}(u_{n}%
)}\right)  \left[  w\right]  \rightarrow0\label{uni}%
\end{equation}

and%
\begin{equation}
d_{\mathbf{A}}I(U_{n})\left[  \mathbf{w}\right]  \rightarrow0\label{dui}%
\end{equation}
where $d_{u}$ and $d_{\mathbf{A}}$ denote the partial differentials of $I$
with respect $u$ and $\mathbf{A}.$ So from (\ref{uni}) we get for any $w\in
C_{0}^{\infty}(\mathbb{R}^{3}\backslash\Sigma)$
\[
d_{u}I(U_{n})\left[  w\right]  -\frac{\sigma^{2}K_{q}^{\prime}(u_{n}%
)}{2\left(  K_{q}(u_{n})\right)  ^{2}}\left[  w\right]  \rightarrow0
\]
which can be written as follows%
\begin{equation}
d_{u}I(U_{n})\left[  w\right]  -\frac{\omega_{n}^{2}K_{q}^{\prime}(u_{n})}%
{2}\left[  w\right]  \rightarrow0\label{duo}%
\end{equation}

where%
\[
\omega_{n}=\frac{\sigma}{K_{q}(u_{n})}.
\]
By Lemma \ref{mi} we have (up to a subsequence)
\[
\omega_{n}\rightarrow\omega_{0}>0
\]
Then by (\ref{duo}) we get for any $w\in C_{0}^{\infty}(\mathbb{R}%
^{3}\backslash\Sigma)$%
\begin{equation}
d_{u}I(U_{n})\left[  w\right]  -\frac{\omega_{0}^{2}K_{q}^{\prime}(u_{n})}%
{2}\left[  w\right]  \rightarrow0.\label{dudu}%
\end{equation}
Now let $\Phi_{n}$ be the solution in $\mathcal{D}^{1,2}$ of the equation
\begin{equation}
-\Delta\Phi_{n}+q^{2}u_{n}^{2}\Phi_{n}=qu_{n}^{2}.\label{fract}%
\end{equation}

Since $\left\{  u_{n}\right\}  $ is bounded in $H^{1}$ (see (\ref{bou}) and
(\ref{bouu})) and since $\Phi_{n}$ solves (\ref{fract}), standard Sobolev
estimates show that $\left\{  \Phi_{n}\right\}  $ is bounded in $\mathcal{D}%
^{1,2}$ and that its weak limit (up to subsequence) $\Phi_{0}$ is a weak
solution of
\begin{equation}
-\Delta\Phi_{0}+q^{2}u_{0}^{2}\Phi_{0}=qu_{0}^{2}.\label{med}%
\end{equation}
Then, by Lemma \ref{unno}, we have%
\begin{equation}
K_{q}^{\prime}(u_{n})=2u_{n}(1-q\Phi_{n})^{2}\text{ and }K_{q}^{\prime}%
(u_{0})=2u_{0}(1-q\Phi_{0})^{2}.\label{util}%
\end{equation}
By standard calculations we have:
\[
\text{for any }w\in C_{0}^{\infty}(\mathbb{R}^{3}\backslash\Sigma)
\]%
\begin{equation}%
{\displaystyle\int}
u_{n}(1-q\Phi_{n})^{2}w\rightarrow%
{\displaystyle\int}
u_{0}(1-q\Phi_{0})^{2}w.\label{biutil}%
\end{equation}
Then, by (\ref{util}) and (\ref{biutil}), we get for any $w\in C_{0}^{\infty
}(\mathbb{R}^{3}\backslash\Sigma)$%
\begin{equation}
K_{q}^{\prime}(u_{n})\left[  w\right]  \rightarrow\text{ }K_{q}^{\prime}%
(u_{0})\left[  w\right]  .\label{triutil}%
\end{equation}
Similar standard estimates show that for any $w\in C_{0}^{\infty}%
(\mathbb{R}^{3}\backslash\Sigma)$%
\begin{equation}
d_{u}I(U_{n})\left[  w\right]  \rightarrow d_{u}I(U_{0})\left[  w\right]
.\label{dado}%
\end{equation}
Then, passing to the limit in (\ref{dudu}), by (\ref{triutil}) and
(\ref{dado}), we get%
\begin{equation}
d_{u}I(U_{0})\left[  w\right]  -\frac{\omega_{0}^{2}K_{q}^{\prime}(u_{0})}%
{2}\left[  w\right]  =0\text{ for any }w\in C_{0}^{\infty}(\mathbb{R}%
^{3}\backslash\Sigma).\label{pe}%
\end{equation}
On the other hand similar arguments show that we can pass to the limit also in
$d_{\mathbf{A}}I(U_{n})\left[  \mathbf{w}\right]  $ and have
\[
\text{for all }\mathbf{w}\in\left(  C_{0}^{\infty}(\mathbb{R}^{3})\right)
^{3}%
\]%
\begin{equation}
d_{\mathbf{A}}I(U_{n})\left[  \mathbf{w}\right]  \rightarrow d_{\mathbf{A}%
}I(U_{0})\left[  \mathbf{w}\right]  .\label{due}%
\end{equation}
From (\ref{dui}) and (\ref{due}) we get%
\begin{equation}
d_{\mathbf{A}}I(U_{0})\left[  \mathbf{w}\right]  =0\text{ for all }%
\mathbf{w}\in\left(  C_{0}^{\infty}(\mathbb{R}^{3})\right)  ^{3}.\label{po}%
\end{equation}
By (\ref{pe}) and (\ref{po}) we deduce, by using density and continuity
arguments, that $U_{0}=\left(  u_{0},\mathbf{A}_{0}\right)  $ is a critical
point of $E_{\sigma_{0},q}$ with $\sigma_{0}=\omega_{0}K_{q}(u_{0})>0.$
\end{proof}

Proof of Theorem \ref{main}

\begin{proof}
The first part of Theorem \ref{main} immediately follows from Propositions
\ref{sol}, \ref{fond} and Theorem \ref{finale}. In fact, if $u_{0}%
,\mathbf{A}_{0}$ are like in Proposition \ref{sol}, by Proposition \ref{fond}
and Theorem \ref{finale} we deduce that $(u_{0},\omega_{0},\phi_{0}%
\mathbf{,A}_{0})$ with $\omega_{0}=$ $\frac{\sigma_{0}}{K_{q}(u_{0})},$
$\phi_{0}=Z_{\omega_{0}}(u_{0})$ solves (\ref{z1}), (\ref{z3}), (\ref{z4}).

Now assume $q=0,$ then, by (\ref{z3}) and (\ref{z4}), we easily deduce that
$\phi_{0}=0$ and $\mathbf{A}_{0}=0.$ Finally assume that $q>0.$ Then, since
$\omega_{0}>0$, by (\ref{z3}) we deduce that $\phi_{0}\neq0.$ Moreover by
(\ref{z4}) we easily deduce that $\mathbf{A}_{0}\neq0$ if and only if
$\ell\neq0.$
\end{proof}

\bigskip

\end{document}